\documentclass[12pt]{article}

\usepackage{amsmath}
\usepackage{amsthm}
\usepackage{amssymb}
\usepackage{cite}
\usepackage{url}
\usepackage{footmisc}
\usepackage{graphicx}
\usepackage{epstopdf}
\usepackage{chngcntr}
\usepackage{hhline}
\usepackage{array}
\usepackage{hyperref}
\usepackage{color}
\usepackage{multirow}
\usepackage{tabularx}
\hypersetup{colorlinks=false}

\usepackage{enumitem}
\setlist[enumerate]{leftmargin=.5in}
\setlist[itemize]{leftmargin=.5in}

\if@twoside
\else 
\fi

\numberwithin{equation}{section}
\counterwithin{figure}{section}
\counterwithin{table}{section}

\theoremstyle{plain}% default
\newtheorem{theorem}{Theorem}[section]

\newtheorem{proposition}[theorem]{Proposition}

\newtheorem{problem}{Problem}

\newtheorem{model}{Model}
\newtheorem{algo}{Algorithm}
\newtheorem{assumption}{Assumption}[section]

\theoremstyle{definition}
\newtheorem{definition}{Definition}[section]

\theoremstyle{remark}
\newtheorem*{remark}{Remark}

\newcommand{\norm}[1]{\left\|#1\right\|}
\newcommand{\abs}[1]{\left\vert#1\right\vert}

\newcommand{\kl}[1]{\left(#1\right)}
\newcommand{\Kl}[1]{\left\{#1\right\}}

\newcommand{\R}{\mathbb{R}} 
\newcommand{\C}{\mathbb{C}}
\newcommand{\N}{\mathbb{N}}
\newcommand{\Z}{\mathbb{Z}}

\newcommand{\F}{\mathcal{F}}
\newcommand{\FI}{\mathcal{F}^{-1}}
\newcommand{\xv}{\boldsymbol{x}}
\newcommand{\xiv}{\boldsymbol{\xi}}

\newcommand{\sv}{\boldsymbol{s}}
\newcommand{\svjk}{\sv_{j,k}}
\newcommand{\svzz}{\sv_{0,0}}

\newcommand{\tv}{\boldsymbol{t}}
\newcommand{\tvpq}{\tv_{p,q}}
\newcommand{\tvzz}{\tv_{0,0}}

\newcommand{\vvjkpq}{\boldsymbol{v}_{j,k,p,q}}

\newcommand{\phijk}{\phi_{j,k}}

\newcommand{\vphi}{\varphi}
\newcommand{\vphipq}{\vphi_{p,q}}

\newcommand{\Tjk}{T_{j,k}}
\newcommand{\Ttpq}{\tilde{T}_{p,q}}

\newcommand{\Itmpq}{\tilde{I}^m_{p,q}}
\newcommand{\Itmzz}{\tilde{I}^m_{0,0}}
\newcommand{\Dtmpq}{\tilde{D}^m_{p,q}}

\newcommand{\Ojk}{\Omega_{j,k}}
\newcommand{\Ozz}{\Omega_{0,0}}
\newcommand{\Otpq}{\tilde{\Omega}_{p,q}}
\newcommand{\Otzz}{\tilde{\Omega}_{0,0}}

\newcommand{\chitzz}{\tilde{\chi}_{0,0}}
\newcommand{\chitpq}{\tilde{\chi}_{p,q}}

\newcommand{\phitpq}{\tilde{\phi}_{p,q}}

\newcommand{\cjk}{c_{j,k}}
\newcommand{\dpq}{d_{p,q}}

\newcommand{\Djk}{\Delta_{j,k}}

\renewcommand{\Im}{I^m}
\newcommand{\Imzz}{I^m_{0,0}}
\newcommand{\Imjk}{I^m_{j,k}}

\newcommand{\ctjkpq}{\tilde{c}_{j,k}(p,q)}

\newcommand{\LtRt}{{L^2(\R^2)}}
\newcommand{\Rt}{{\R^2}}
\newcommand{\LiRt}{{L^\infty(\R^2)}}

\newcommand{\etajk}{\eta_{j,k}}
\newcommand{\etazz}{\eta_{0,0}}
\newcommand{\etatjkpq}{\tilde{\eta}_{j,k,p,q}}
\newcommand{\etatzzzz}{\tilde{\eta}_{0,0,0,0}}
\newcommand{\gapq}{\gamma_{p,q}}
\newcommand{\gazz}{\gamma_{0,0}}
\newcommand{\epsjk}{\varepsilon_{j,k}}
\newcommand{\dejk}{\delta_{j,k}}
\newcommand{\dezz}{\delta_{0,0}}
\newcommand{\zetapq}{\zeta_{p,q}}
\newcommand{\zetazz}{\zeta_{0,0}}

\newcommand{\Ejk}{\mathcal{E}_{j,k}}
\newcommand{\Etpq}{\tilde{\mathcal{E}}_{p,q}}

\title{Subaperture-based Digital Aberration Correction for OCT: A Novel Mathematical Approach}
\author{
	Simon Hubmer\footnote{Johann Radon Institute Linz, Altenbergerstra{\ss}e~69, 4040 Linz, Austria, (simon.hubmer@ricam.oeaw.ac.at), Corresponding author.} ,
	Ekaterina Sherina\footnote{University of Vienna, Faculty of Mathematics, Oskar Morgenstern-Platz 1, 1090 Vienna, Austria (ekaterina.sherina@univie.ac.at)} ,
	Ronny Ramlau\footnote{Johannes Kepler University Linz, Institute of Industrial Mathematics, Altenbergerstra{\ss}e~69, 4040 Linz, Austria, (ronny.ramlau@jku.at)} \footnote{Johann Radon Institute Linz, Altenbergerstra{\ss}e~69, 4040 Linz, Austria, (ronny.ramlau@ricam.oeaw.ac.at)} , 
	\\
	Michael Pircher\footnote{Medical University of Vienna, Center for Medical Physics and Biomedical Engineering, W\"ahringer G\"urtel~18-20, 1090 Vienna, Austria, (michael.pircher@meduniwien.ac.at)} ,
	Rainer Leitgeb\footnote{Medical University of Vienna, Center for Medical Physics and Biomedical Engineering, W\"ahringer G\"urtel~18-20, 1090 Vienna, Austria, (rainer.leitgeb@meduniwien.ac.at)}}
\begin{document}

% Include the title
\maketitle

% Abstract
\begin{abstract}
In this paper, we consider subaperture-based approaches for the digital aberration correction (DAC) of optical coherence tomography (OCT) images. In particular, we introduce a mathematical framework for describing this class of approaches, leading to new insights for the subaperture-correlation method. Furthermore, we propose a novel DAC approach requiring only minimal statistical assumptions on the spectral phase of the scanned object. Finally, we demonstrate the applicability of our novel DAC approach via numerical examples based on both simulated and experimental OCT data.


\end{abstract}

% % % % % % % % % % % % %
% Start of the sections %
% % % % % % % % % % % % %

% % % % % % % % % % % % % %
% Section - Introduction  %
% % % % % % % % % % % % % %
\section{Introduction}\label{sect_introduction}

Optical Coherence Tomography (OCT) has found numerous applications in medical imaging, most importantly in ophthalmology \cite{Drexler_Fujimoto_2015,Leitgeb_Placzek_Rank_Krainz_Haindl_Li_Liu_Unterhuber_Schmoll_Drexler_2021}. Successful functional extensions of OCT are, e.g., OCT angiography \cite{Chen_Wang_2017, Leitgeb_Baumann_2018}, phase sensitive OCT 
\cite{Pircher_Hitzenberger_SchmidtErfurth_2017,deBoer_Hitzenberger_Yasuno_2017}, or OCT elastography \cite{Schmitt_1998, Zaitsev_2021, Krainz_Sherina_Hubmer_Liu_Drexler_Scherzer_2022}. OCT provides three-dimensional tissue imaging non-invasively with high axial resolution and excellent structural contrast. The lateral resolution of OCT is decoupled from the axial one, which is determined purely by the partial coherence properties of the light source, whereas the lateral resolution is determined by the numerical aperture of the imaging optics. High spatial resolution is of prime interest for early diagnosis of ocular diseases, when treatment is still most effective and often irreversible functional impairment can be avoided. E.g., first successful gene therapies have been introduced for retinal diseases, as well as cell replacement therapies. All of them call for three-dimensional cellular resolution imaging capability. In case of retinal imaging, cellular resolution has been demonstrated with adaptive optics (AO) assistance, overcoming the substantial aberrations induced by ocular media \cite{Godara_Dubis_Roorda_Duncan_Carroll_2010}. AO requires the measurement of wavefront aberrations using a wavefront sensor. Those reconstructed aberrations are then used to pre-shape the wavefront using a deformable mirror leading to diffraction limited imaging performance at the ocular fundus. Combining AO with OCT allows for high isotropic resolution of the retina in all three dimensions \cite{Pircher_Zawadzki_2017}. Despite first successful demonstration of AO OCT for retinal diagnostics, additional hardware costs, complex maintenance, and restricted access to AO OCT systems have so far prevented a broader acceptance of the technology. Nevertheless, first clinical prototypes have been introduced recently \cite{Shirazi_Andilla_Lefaudeux_etal_2022}.

Recent demonstrations of high resolution retinal imaging based on OCT signal post-processing alone therefore seemed to be an attractive and cost-effective alternative \cite{Liu_South_Xu_ScottCarney_Boppart_2017,Hillmann_Spahr_Hain_Sudkamp_FrankePfaeffle_Winter_Huettmann_2016,Shemonski_South_Liu_Adie_Scott_Boppart_2015}. Those approaches make use of the complex nature of the OCT signal. The phase information of the OCT signal gives in fact also access to the wavefront aberrations. Various methods have been proposed to extract the wavefront aberrations, being termed digital or computational adaptive optics (DAO, CAO) \cite{Kumar_Drexler_Leitgeb_2013,Adie_Graf_Ahmad_Carney_Boppart_2012} as well as digital aberration correction (DAC). In the following, we will adapt the notion of DAC. In short, the aim of DAC is to find the wavefront aberrations, and to correct recorded OCT image planes by applying a phase conjugated filter function to their pupil or Fourier plane. Subsequent inverse Fourier transformation of the modified pupil plane image leads to an aberration corrected image. This procedure needs to be repeated for each depth plane. Alternatively, if the sample is sufficiently homogeneous, aberration correction only needs to be applied to one plane, the other planes being corrected for known defocus. DAC methods can be roughly categorized into iterative correction methods and non-iterative methods. Iterative methods assume polynomial expansion \cite{Liu_South_Xu_ScottCarney_Boppart_2017}, e.g., by Zernike polynomials, for the correction phase and determine the set of coefficients leading to optimal image sharpness. Non-iterative methods are inspired by the principle of Shack-Hartman wavefront sensing, by determining local wavefront slopes in the pupil plane and by fitting polynomials to estimate the wavefront aberrations. The first of such methods, termed subaperture DAC \cite{Kumar_Drexler_Leitgeb_2013,Kumar_Kamali_Platzer_Unterhuber_Drexler_Leitgeb_2015}, splits the pupil plane into a matrix of subdomains and determines the local wavefront slope for each subdomain. A variant of this method uses digital lateral shearing of the pupil phase to retrieve the local wavefront slope for each pixel (DLS DAO) \cite{Kumar_Georgiev_Salas_Leitgeb_2021}. Both methods have the advantage of being independent of system parameters such as the sample refractive index or of the imaging setup. In principle, DLS DAO provides better sampling and is therefore better suited for higher order aberrations, but it requires specular reflections in the image plane to be present for proper performance.  

In the present paper, we therefore elaborate a mathematical framework for subaperture DAC, which allows both to quantify the limitations of the method as well as leading to strategies to further enhance its performance. In particular, we derive a novel DAC approach which combines the principles of subaperture DAC with a physically motivated averaging.

The outline of this paper is as follows: In Section~\ref{sect_background}, we recall the general physical background of OCT and its image formation process. In Section~\ref{sect_Sub_DAC}, we derive a mathematical formalism for general subaperture-based DAC and consider its implications on SC-DAC. In Section~\ref{sect_novel_DAC}, we propose a novel DAC approach based on our derived formalism, which in Section~\ref{sect_numerics} we test on both simulated and experimental data, before ending with a short conclusion in Section~\ref{sect_conclusion}.

% % % % % % % % % % % % % % % % %
% Section - Physical Background %
% % % % % % % % % % % % % % % % %
\section{Physical Background}\label{sect_background}

In this section, we review the physical background on swept source based OCT, as well as on the Shack-Hartmann wavefront sensor.

% Subsection - Swept Source Based OCT
\subsection{Swept Source Based Optical Coherence Tomography}

A typical swept source based OCT system for imaging the retina is depicted in Figure~\ref{fig_SS_OCT_scheme}, (left) with the reconstructed volume being composed of depth- or A-scans at each lateral position $(x,y)$ (right). The volume allows to extract a plane at a given depth referred to as an OCT enface image. In short, the light from a tuneable light source (swept source) is coupled into a single mode fiber based interferometer, where it is split by the first fiber optic beam splitter (FBS1) into sample and reference arm light, respectively. In the sample arm, the light exits the fiber into free space where it is reflected by a scanner (for $x$- and $y$- scanning of the sample). The telescope in the sample arm (lenses L2 and L3) is used to image the pivot point of the scanner onto the pupil plane of the eye for scanning the imaging beam across the retina. The light that is backscattered by the retina traverses the optical media where aberrations are introduced, is then de-scanned by the scanner and back-coupled into the single mode fiber by the fiber-collimator lens L1. In the reference arm the light is coupled out via fiber-collimator lens L4, traverses a variable free-space optical delay line that allows matching the path length difference between sample and reference arm, is coupled back into a single mode fiber via fiber-collimator lens L5, and is brought to interference with light returning from the sample arm in the last fiber beam splitter (FBS2). 

\begin{figure}[ht!]
	\centering
	\includegraphics[width=\textwidth]{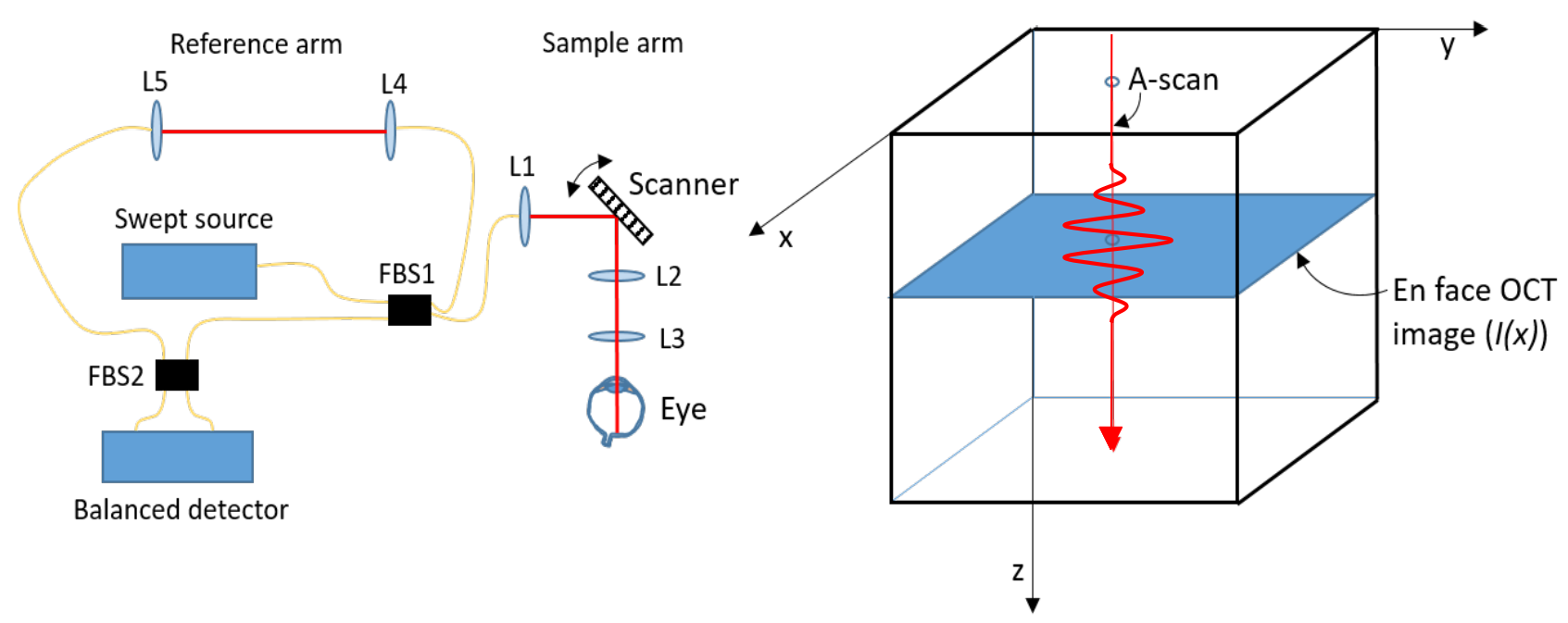}
	\caption{Scheme of a typical optical layout of a swept source based OCT system for retinal imaging (left). L1-L5…lenses. FBS1-2…fiber based beam splitter. Coordinate system used in this work (right)}
	\label{fig_SS_OCT_scheme}
\end{figure}

\begin{figure}[ht!]
	\centering
	\includegraphics[width=\textwidth]{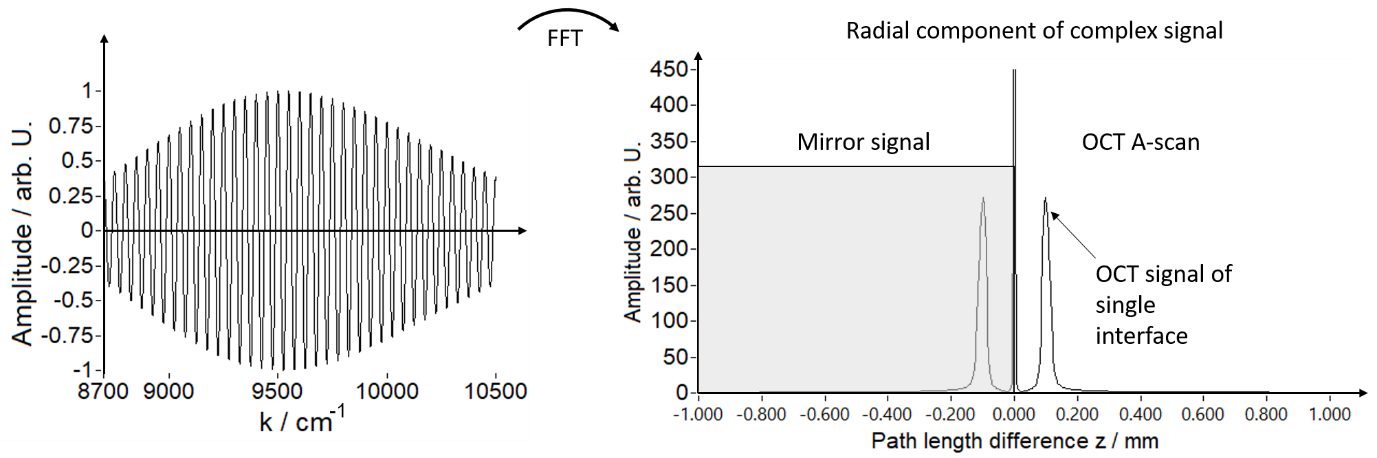}
	\caption{Principle of an A-scan generation in OCT of a single interface located at depth position $0.1$mm. The wavelength of the tuneable light source is rapidly changed over the entire spectral bandwidth of the source (from $950$nm to $1150$nm) and the interferometric signal is recorded by the data acquisition board (left). A depth profile or A-scan is retrieved by fast Fourier transformation of the recorded interferogram (in $k$-space) and by displaying the radial component of the complex valued signal (right). The mirror signal in the grey shaded area (corresponding to negative path length differences) is usually neglected.}
	\label{fig_OCT_Ascan_principle}
\end{figure}

The interferometric signal is then detected by the balanced detector and digitized by a data acquisition board. To generate a depth profile or an A-scan in OCT, the light source is rapidly tuned (swept) over the spectral wavelength band (or wavenumbers $k$) and the corresponding interferometric signal is recorded simultaneously (cf.\ Figure~\ref{fig_OCT_Ascan_principle} left). Successive fast Fourier transformation of this interferometric signal results in the final A-scan (cf.\ Figure~\ref{fig_OCT_Ascan_principle} right). Positive or negative path length differences cannot be distinguished which causes the appearance of a mirror term (cf.\ grey zone in Figure~\ref{fig_OCT_Ascan_principle} right) that is usually neglected by displaying only the positive path length range. For the subsequent discussion, we assume that the wavelength sweep is linear in time and in wavenumber $k (= 1/\lambda)$ space (cf.\ Figure~\ref{fig_OCT_Ascan_principle} left). Note that in the general case, the recorded spectral interference signal might first need to be linearized in $k$, also taking care of possible dispersion unbalance in the interferometer\cite{Wojtkowski2004}. Assuming a single interface in the sample arm with reflectivity $R_S$, the detector current $I$, depending on the time $t$ or wavenumber $k$ linearly tuned in time $t$, can be expressed as \cite{Lexer_Hitzenberger_Fercher_Kulhavy_1997, Choma_Sarunic_Yang_Izatt_2003}
\begin{equation}\label{eq_OCT_1}
	I(k(t)) = C \kl{R_R + R_S(z) + 2 \sqrt{R_R R_S(z)} \cos\kl{2k(t)z+\varphi}} \,,
\end{equation}
where $C$ is a constant determined by the experimental setup (detector responsivity and illumination power), $R_R$ is the reflectivity in the reference arm, $k$ the wavenumber, $z$ denotes the optical path length difference between sample and reference arm, and $\varphi$ is the interferometric phase offset. Fourier transformation of \eqref{eq_OCT_1} with respect to k yields a depth profile or A-scan:
\begin{equation*}\label{eq_OCT_2}
	A(z) = \int I(k(t)) e^{-2\pi i k(t) z} \, dk\,,
\end{equation*}
where $A(z)$ represents a complex valued function. To generate an en-face OCT image ($x$-$y$ imaging plane) at given depth $z$, A-scans are recorded at several lateral positions of the sample and the amplitude and phase of $A(z)$ are retrieved at the position $z$. This complex valued en-face OCT image resulting from the measurement is referred to in the following as $I^m(\xv)$, the superscript $m$ indicating that the experimentally obtained en-face OCT image may additionally be affected by wavefront aberrations.

% Subsection - The Shack-Hartmann Wavefront Sensor
\subsection{The Shack-Hartmann Wavefront Sensor}\label{subsect_SH_WFS}

In this section, we review some basic properties of the Shack-Hartmann wavefront sensor (SH-WFS) \cite{Ellerbroek_Vogel_2009,Platt_Shack_2001,Primot_2003}, which provides the motivation of the subaperature-based digital aberration correction methods introduced below. First, note that in general the phase or wavefront of light cannot be measured directly, but can be determined from indirect measurements. For example, we saw above that in OCT the phase is encoded in the interference pattern between the sample and the reference field, from which it can be computed; cf.~\cite{Drexler_Fujimoto_2015,Elbau_Mindrinos_Veselka_2023}. Among the many different types of wavefront sensors, the SH-WFS described below is perhaps the most commonly used one. A popular alternative is the pyramid wavefront sensor, where the phase is reconstructed from the intensity patterns resulting from light being focused onto the tip of a pyramidal prism \cite{Ragazzoni_1996, Brunner_Shatokhina_Shirazi_Drexler_Leitgeb_Pollreisz_Hitzenberger_Ramlau_Pircher_2021}.

% % % % % % % % % % % % % % % %
% Figure - Shack-Hartmann WFS %
% % % % % % % % % % % % % % % %
\begin{figure}[ht!]
	\centering
	\includegraphics[width=\textwidth]{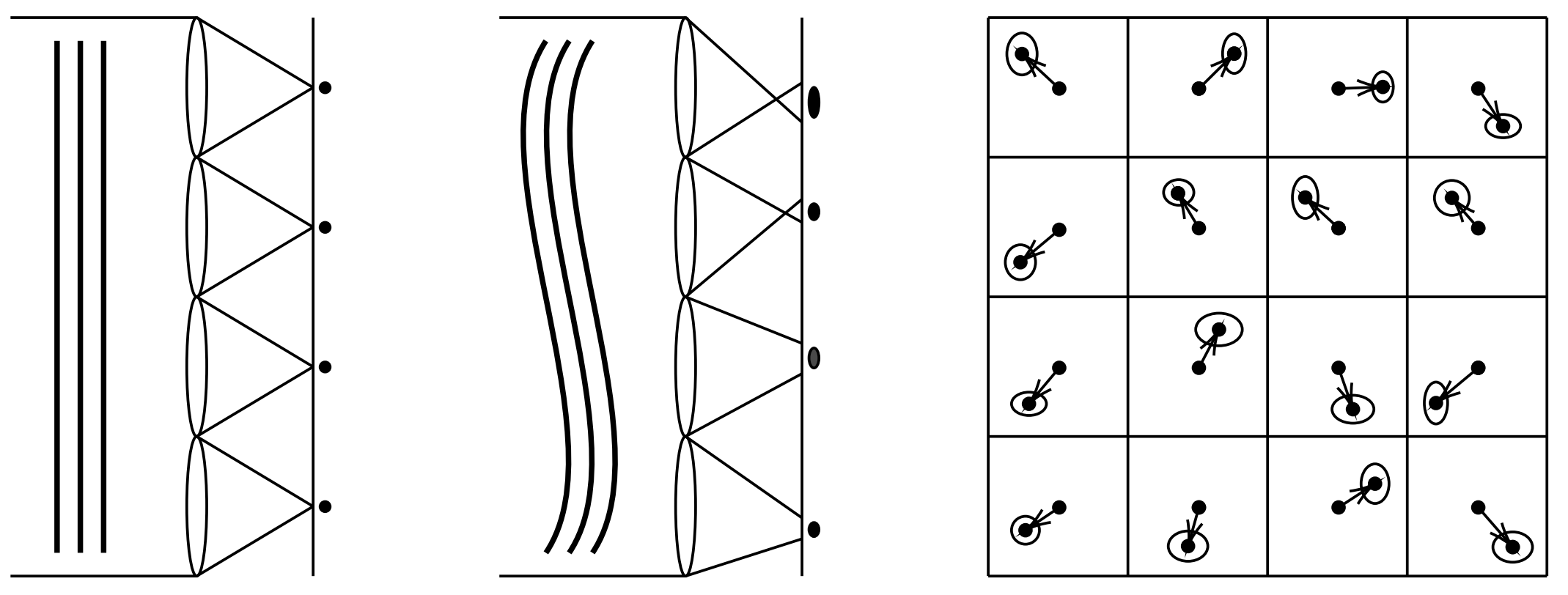}
	\caption{Schematic depiction of the physical principle behind the SH-WFS. Light from a distant point source with a plane (non-aberrated) wavefront is focused by lenslets into centered focal spots (left). If the wavefront is not plane (i.e.\ aberrated), the focal spots are blurred and shifted with respect to their central locations (middle). At the CCD detector, these shifts are detected via a comparison of the centers of mass (right).}
	\label{fig_Shack_Hartmann_WFS}
\end{figure}

The principle of the SH-WFS, schematically depicted in Figure~\ref{fig_Shack_Hartmann_WFS}, can be summarized as follows: The incoming light passes through a quadratic array of small lenses, also called lenslets, which focus the light onto a CCD photon detector, subdivided into a regular grid of so-called subapertures. If the incident light stems from a distant point source with a plane (non-aberrated) wavefront, then each lenslet focuses the light into a focal spot at the center of the corresponding subaperture. However, if the wavefront is not plane (i.e., aberrated), then the focal spots are blurred and shift away from the center of the subapertures. These shifts, typically computed with respect to the centers of mass of the focal spots \cite{Gilles_Ellerbroek_2006}, are measured from the CCD data, and are related to the average gradient (slope) of the incoming wavefront over each subaperture. Hence, from the measured shifts one can reconstruct an approximation of the incoming wavefront, which is itself a non-trivial inverse problem \cite{Roddier_1999} for which many different reconstruction approaches have been proposed. In this paper, we use the fast, stable, and parallelizable Cumulative Reconstructor with
Domain Decomposition (CuReD) 
\cite{Zhariy_Neubauer_Rosensteiner_Ramlau_2011,Rosensteiner_2011_01,Rosensteiner_2011_02}.

% % % % % % % % % % % % % % % % % % % % % % % % % % % % % % %
% Section - Subaperture-based Digital Aberration Correction %
% % % % % % % % % % % % % % % % % % % % % % % % % % % % % % %
\section{Subaperture-based Digital Aberration Correction}\label{sect_Sub_DAC}

In this section, we derive a mathematical formalism for subaperture-based DAC approaches. For this, we first recall some basic results on the 2D Fourier transform, and then introduce a mathematical model which connects complex-valued OCT images to wavefront aberrations. Using this model, we then continue by formalizing the subdomains, subapertures, and subimages considered in subaperture-based DAC approaches, which we then use to formalize and review the previously proposed SC-DAC method. 

Note that in the following we assume that OCT images are complex-valued functions on $\R^2$, which are in practice represented by complex-valued images on a pixel grid.

% % % % % % % % % % % % % % % % % % % % % %
% Subsection - Mathematical Preliminaries %
% % % % % % % % % % % % % % % % % % % % % %
\subsection{Mathematical Preliminaries}\label{subsect_preliminaries}

The Fourier transform $\F(f)$ of an integrable function $f: \R^2 \to \C$ is defined by
\begin{equation*}
	\F(f)(\xi_1,\xi_2) :=  \int_{-\infty}^{\infty} f(x_1,x_2) e^{-2\pi i \kl{x_1,x_2} \cdot \kl{\xi_1,\xi_2}  } \, d (x_1,x_2) \,.
\end{equation*}
Together with $\xv := (x_1,x_2)$ and $\xiv := (\xi_1,\xi_2)$ this can be written in the compact form 
\begin{equation*}
	\F(f)(\xiv) = \int_{-\infty}^{\infty} f(\xv) e^{-2\pi i \xv \cdot \xiv } \, d \xv \,.
\end{equation*}
It is well-known that the Fourier transform can be continuously extended to a bounded linear operator $\F : \LtRt \to \LtRt$ and that it is continuously invertible by
\begin{equation*}
	\FI( g )(\xv) := \int_{-\infty}^{\infty} g(\xiv) e^{2\pi i \xiv \cdot \xv } \, d \xiv \,.
\end{equation*}
Two useful properties of the Fourier transform are the \emph{translation property}
\begin{equation}\label{prop_translation}
	h(\xv) := f(\xv - \xv_0) 
	\quad
	\Longrightarrow
	\quad
	\F (h)(\xiv) = e^{-2\pi i \xv_0 \cdot \xiv} \F(f)(\xiv) \,,
\end{equation}
as well as the \emph{modulation property}
\begin{equation}\label{prop_modulation}
	h(\xv) := e^{2\pi i \xv \cdot \xiv_0 } f(\xv) 
	\quad
	\Longrightarrow
	\quad
	\F (h)(\xiv) = \F(f)(\xiv - \xiv_0) \,,
\end{equation}
which we will make repeated use of in the upcoming analysis. For proofs of these properties as well as further details we refer to the standard literature on Fourier transforms.

% % % % % % % % % % % % % % % % % % % % % % % % % % % %
% Subsection - Mathematical Model of Image Formation  %
% % % % % % % % % % % % % % % % % % % % % % % % % % % %
\subsection{Mathematical Model of Image Formation}

In this section, we derive a mathematical model describing the influence of aberrations on the (virtual) image formation process in OCT. The model connects the quantities:
\begin{itemize}
	\item $\Im = \Im(\xv)$ ... the measured complex-valued OCT image, i.e., with aberrations,
	\item $I = I(\xv)$ ... the aberration-free complex-valued OCT image,
	\item $\phi = \phi(\xiv)$ ... the wavefront aberration phase, a real-valued function.
\end{itemize}    
For introducing our mathematical model, it is convenient to make the following
\begin{definition}
	Let $I : \R^2 \to \C$ and $\Im : \R^2 \to \C$. Then for all $\xiv \in \R^2$ we define the Fourier coefficients of $I$ and $\Im$ by
	\begin{equation*}
		D := D(\xiv) := \F(I)(\xiv) \,,
		\qquad
		\text{and}
		\qquad
		D^m := D^m(\xiv) := \F(\Im)(\xiv) \,. 
	\end{equation*}
	\label{def_ddm}
\end{definition}

% % % % % % % % % % % % % % % % % %
% Figure - Image Formation Model  %
% % % % % % % % % % % % % % % % % %
\begin{figure}[ht!]
	\centering\includegraphics[width=\textwidth]{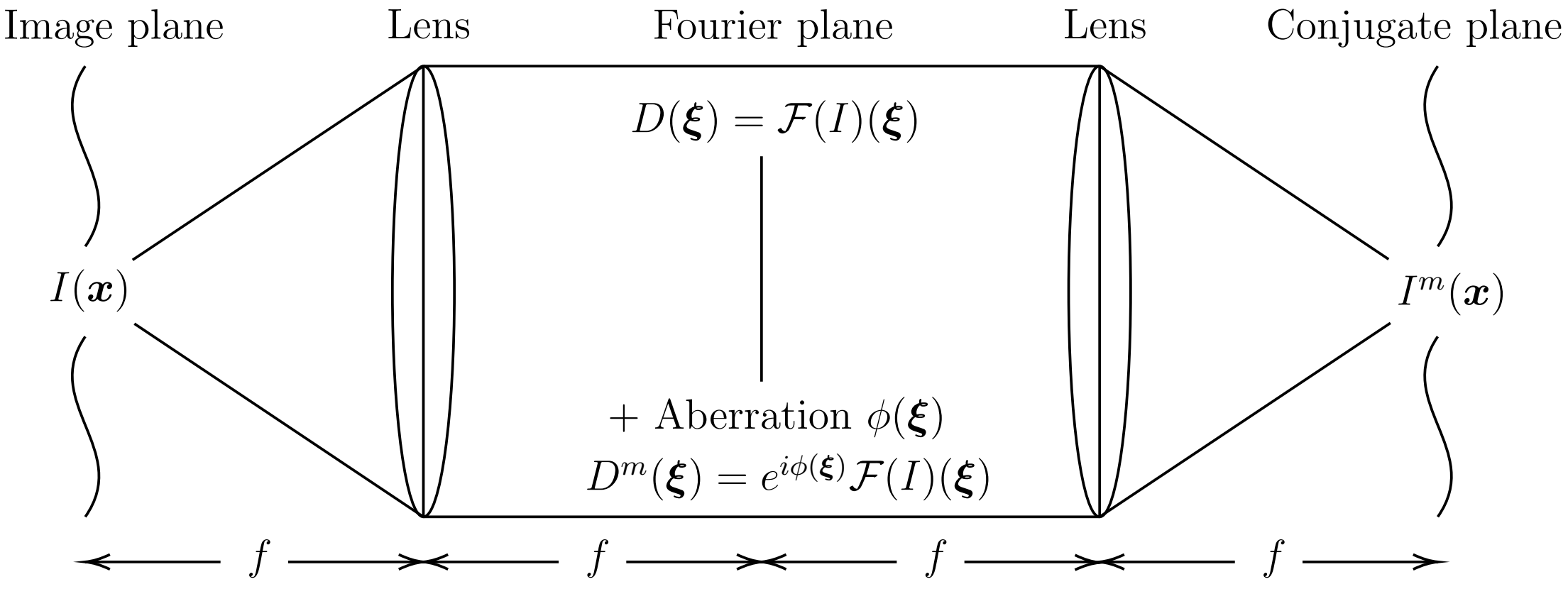}
	\caption{Schematic depiction of the virtual image formation model \eqref{model}.}
	\label{fig_image_formation_model}
\end{figure}

As discussed in Section~\ref{sect_background}, wavefront aberrations emerge throughout the OCT scanning system. Hence, one way to derive a complete mathematical model of the connection between OCT images and aberrations is to combine the different components in an OCT system into a single mathematical model. However, this is undesirable, since it would require the model to be adapted for every different OCT scanning setup and in general would also lead to a highly complex model. Hence, in this paper we opt for the following simpler model: It is known \cite{Goodman_2005} that mathematically every imaging system can be simplified to a system containing two properly aligned lenses; cf.~Figure~\ref{fig_image_formation_model}. In this paper, we assume that wavefront aberrations occur predominantly as phase shifts in the Fourier plane; again cf.~Figure~\ref{fig_image_formation_model}. Mathematically, this can be expressed as
\begin{equation}\label{eq_Dm_D}
	D^m(\xiv) = e^{i \phi(\xiv)} D(\xiv) \,.
\end{equation}    
Hence, we obtain the following mathematical model:
\begin{model}\label{model_main}
	The complex-valued functions $\Im \in \LtRt$ and $I \in \LtRt$ representing the measured and the aberration-free OCT images, respectively, are assumed to be connected via
	\begin{equation}\label{model}
		\Im(\xv) = \FI\kl{ e^{i \phi(\xiv)} \F(I)(\xiv) }(\xv) \,, 
		\qquad
		\forall \, \xv \in \R^2 \,,
	\end{equation}
	where $\phi \in \LtRt$ is a real-valued function representing the wavefront-aberrations.
\end{model}

Note that by rearranging the model equation \eqref{model} we obtain
\begin{equation}\label{image_correction}
	I(\xv) = \FI\kl{ e^{- i \phi(\xiv)} \F\kl{\Im}(\xiv) }(\xv) \,,
	\qquad
	\forall \, \xv \in \R^2 \,.
\end{equation}
Hence, given an estimate of the wavefront aberration $\phi$, the aberration-free image $I$ can be reconstructed from the measured image $\Im$. The aim of DAC is thus to find an estimate of $\phi$ and use it to correct $\Im$ via \eqref{image_correction}. Hence, the DAC problem can be summarized as follows: 

\begin{problem}[DAC]\label{problem_DAC}
Given a measured, complex-valued OCT image $\Im$ and the Model~\ref{model_main} connecting it to the aberration-free OCT image $I$, the general DAC problem consists in determining the wavefront aberration $\phi$ and using it to correct the measured image $\Im$ via \eqref{image_correction}.
\end{problem}

As can be seen from \eqref{model}, solving the general DAC Problem~\ref{problem_DAC} is impossible without imposing additional assumptions, since the measured image $\Im$ depends both on the aberration-free image $I$ and the wavefront aberration $\phi$. However, this is also the case in hardware-based aberration correction, since the spectral phase of the scanned object and the wavefront aberration interact; cf.~Section~\ref{sect_background} and in particular Section~\ref{subsect_novel_DAC_motivation} below. Nevertheless, wavefront sensors such as the SH-WFS are successfully used for image correction in OCT. Hence, in the following we consider the general approach of subaperture-based DAC for solving Problem~\ref{problem_DAC}, which is based on the physical behaviour of the SH-WFS; cf.~Section~\ref{subsect_SH_WFS}. This approach is based on a subdivison of the Fourier domain into subdomains modelling the subapertures of a SH-WFS, and comparing the resulting images to estimate the wavefront aberration $\phi$. As for the SH-WFS, the choice of the number and size of these subapertures depends on the general imaging conditions and on resolution restrictions, which we will return to in Section~\ref{sect_numerics}.

% % % % % % % % % % % % % % % % % % % % % % % % % % % %
% Subsection - Subdomains/Subapertures and Subimages  %
% % % % % % % % % % % % % % % % % % % % % % % % % % % %
\subsection{Subdomains/Subapertures and Subimages}\label{subsect_sub_01}

In this section, we introduce various quantities required for subaperture-based DAC approaches. These can be considered purely from a mathematical point of view, but they also have a motivation stemming from the SH-WFS, see Section~\ref{subsect_SH_WFS}, which we now present before giving precise mathematical definitions below. In particular, we define subdomains $\Ojk$ which form a regular rectangular tiling of $\R^2$ and correspond to the subapertures of a SH-WFS, as well as shift operators $\Tjk$ translating points in the central subdomain $\Omega_{0,0}$ to corresponding points in the subdomains $\Ojk$; cf.~Figure~\ref{fig_subdomains}. Furthermore, we consider the indicator functions $\chi_{j,k}$ of these domains, and use them to define subimages $\Imjk$ corresponding to the images which are obtained when light passes through individual lenslets of a SH-WFS; cf.~Figure~\ref{fig_subimage_definition}. Finally, we define the local wavefront aberrations $\phijk$ as the restrictions of the wavefront aberration $\phi$ to the different subdomains $\Ojk$. The average slopes of these local aberrations $\phijk$ are what is measured by the SH-WFS.

All of the above quantities are required for the mathematical formulation and analysis of the subaperture-based DAC approaches considered below. Hence, we now start by defining subdomains representing (virtual) subapertures; see also Figure~\ref{fig_subdomains}.

\begin{definition}\label{def_subdomains}
	For $J, K \in \R^+$ we define the subdomain/subaperture
	\begin{equation*}
		\Omega_{0,0} := [ -J/2, J/2 ) \times [-K/2, K/2)  \subset \R^2 \,,
	\end{equation*}
	as well as the shifted subdomains/subapertures
	\begin{equation*}
		\Ojk := \Tjk(\Omega_{0,0}) 
		= 
		\Kl{\xiv \in \R^2 \, \vert \, \Tjk^{-1}(\xiv) \in \Omega_{0,0}} \,, 
		\qquad
		\forall \, j,k\in \Z \,,
	\end{equation*}
	where the translation operator $\Tjk$ is defined by
	\begin{equation}\label{def_Tjk}
		\Tjk(\xiv) := (\xi_1 + jJ, \xi_2 + k K) \,,
		\qquad
		\forall \, \xiv \in \R^2 \,, \quad \forall \, j,k \in \Z   \,.
	\end{equation}
\end{definition}

% % % % % % % % % % % % % % % % % % % % % % 
% Figure - Subdomains and Shift Operator  %
% % % % % % % % % % % % % % % % % % % % % %
\begin{figure}[ht!]
	\centering
	\includegraphics[width=\textwidth]{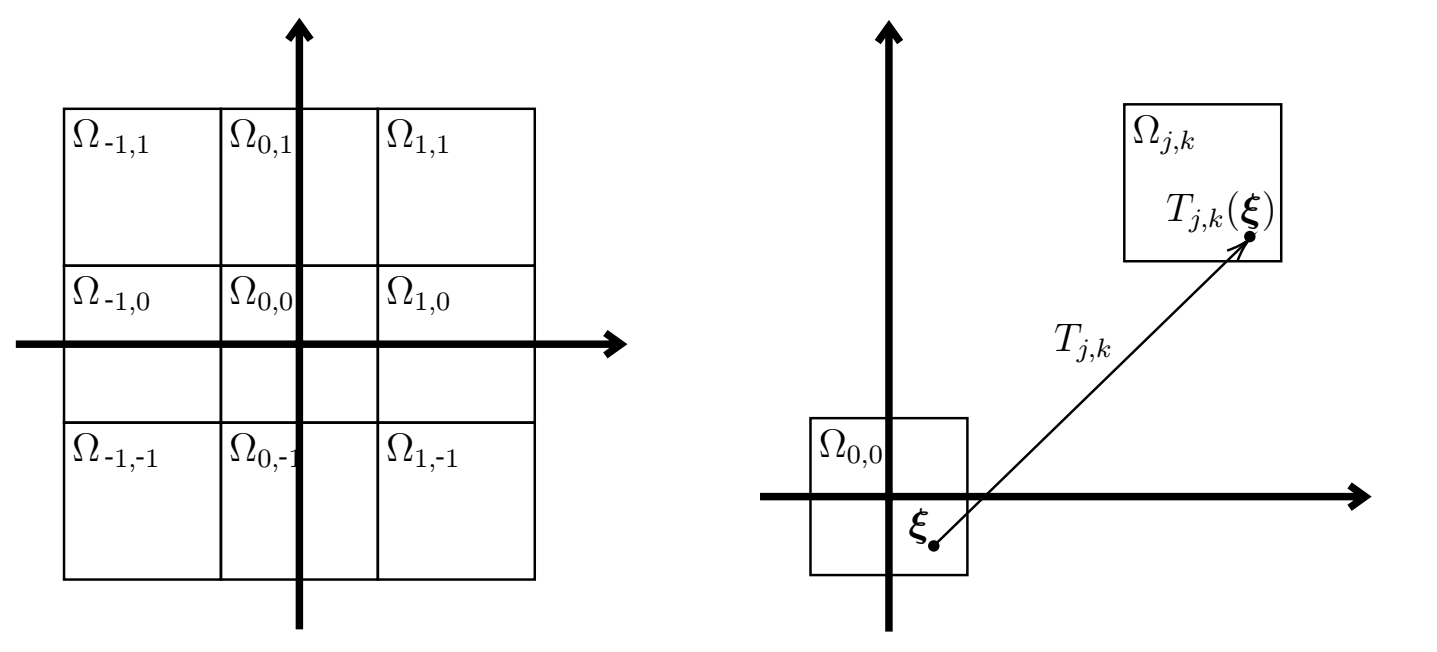}
	\caption{Illustration of subdomains $\Ojk$ (left) and shift operator $\Tjk$ (right).}
	\label{fig_subdomains}
\end{figure}

As the name suggests, we now use the subdomains $\Ojk$ as representations of (virtual) subapertures in the Fourier plane for splitting up the OCT image $I^m$ into different (virtual) subimages $\Imjk$. More precisely, for all $j,k\in\Z$ and $\xiv \in \R^2$ let
\begin{equation*}
	\chi_{j,k}(\xiv) 
	:=
	\begin{cases}
		1 \,, & \xiv \in \Ojk \,,
		\\
		0 \,, & \text{else} \,,
	\end{cases}
\end{equation*}
denote the indicator functions of the subdomains $\Ojk$. Then the effect of a subaperture in the Fourier plane on the Fourier coefficients $D^m$ precisely amounts to the ``cut-outs'':
\begin{equation}\label{def_Dmjk}
	D^m_{j,k}(\xiv) :=  \chi_{j,k}(\xiv) D^m(\xiv)  \overset{\eqref{eq_Dm_D}}{=} \chi_{j,k}(\xiv) e^{i \phi(\xiv)} \F(I)(\xiv) \,,
	\qquad
	\forall \, \xiv \in \R^2 \,.
\end{equation}
The corresponding (virtual) subimages can now be defined as follows:
\begin{definition}
	For each $j,k \in \Z$ the subimages $\Imjk$ are defined via
	\begin{equation}\label{def_Imjk}
		\Imjk(\xv) := \FI( D^m_{j,k}( \Tjk(\xiv)))(\xv)\,,
		\qquad
		\forall \xv \in \R^2 \,,
	\end{equation}
	where the cut-out Fourier data $D^m_{j,k}$ are given by \eqref{def_Dmjk}. 
\end{definition}

% % % % % % % % % % % % % % % % %
% Figure - Subimage Definition  %
% % % % % % % % % % % % % % % % %
\begin{figure}
	\centering
	\includegraphics[width=\textwidth]{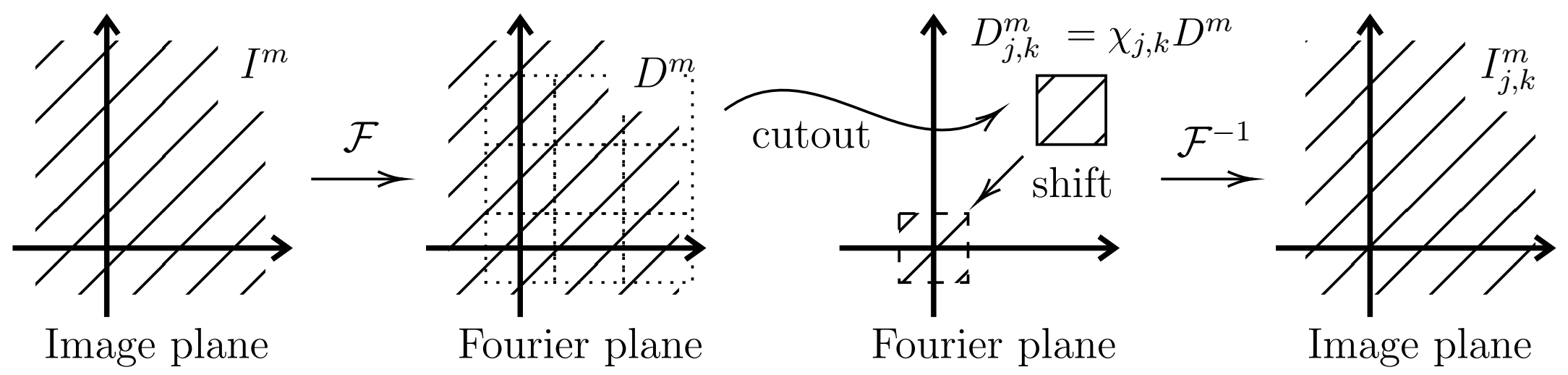}
	\caption{Illustration of the definition of the subimages $\Imjk$ given in \eqref{def_Imjk}.}
	\label{fig_subimage_definition}
\end{figure}

For a better understanding, the definition of the subimages $\Imjk$ is illustrated in Figure~\ref{fig_subimage_definition}. Their definition mimicks the image formation process of light being focused through the different lenslets of the SH-WFS. Note that the translation operator $\Tjk(\xiv)$ appears in \eqref{def_Dmjk} to model the required shift of the Fourier coefficients to the center of the domain. The above definition implies

\begin{proposition}\label{prop_subimages}
	Let the subimages $\Imjk$ be defined as in \eqref{def_Imjk} and let 
	\begin{equation}\label{def_phijk}
		\phijk(\xiv) := \phi(\Tjk(\xiv)) \,,
		\qquad
		\forall \, \xiv \in \Omega_{0,0} \,, \quad
		\forall \, j,k \in \Z \,,
	\end{equation}
	be the local wavefront aberrations. Then for all $\xv \in \R^2$ we obtain the representation
	\begin{equation}\label{Imjk_exp}
		\Imjk(\xv) = \FI\kl{\chi_{0,0}(\xiv) e^{i \phijk(\xiv)} \F(I)(\Tjk(\xiv)) }(\xv) \,.
	\end{equation}
\end{proposition}
\begin{proof}
	Due to the definitions \eqref{def_Dmjk} and \eqref{def_Imjk} of $D^m_{j,k}$ and $\Imjk$, respectively, we have
	\begin{equation*}
		\Imjk(\xv) = \FI(\chi_{j,k}(\Tjk(\xiv)) e^{i \phi(\Tjk(\xiv))} \F(I)(\Tjk(\xiv)) )(\xv) \,.
	\end{equation*}
	Now, since by the definition of $\chi_{j,k}$ we have
	\begin{equation*}
		\chi_{j,k}(\Tjk(\xiv)) = \chi_{0,0}(\xiv) \,,
	\end{equation*}
	it follows that 
	\begin{equation}\label{helper_1}
		\Imjk(\xv) = \FI(\chi_{0,0}(\xiv) e^{i \phi(\Tjk(\xiv))} \F(I)(\Tjk(\xiv)) )(\xv) \,,
	\end{equation}
	which together with the definition \eqref{def_phijk} of $\phijk$ yields the assertion.
\end{proof}

Note that if there are no aberrations, i.e., if $\phi(\xiv) \equiv 0$, then \eqref{Imjk_exp} implies
\begin{equation*}
	\Imjk(\xv) = \FI\kl{\chi_{0,0}(\xiv) \F(I)(\Tjk(\xiv)) }(\xv)  \,.
\end{equation*}
Hence, for the central subimage, i.e., for $j=k=0$, we obtain
\begin{equation*}
	\Imzz(\xv) = \FI\kl{\chi_{0,0}(\xiv)  \F \kl{I}(\xiv)}(\xv) \,,
\end{equation*}
which is a low-resolution version of the original image $I$. For all other $j,k\neq 0$, the presence of the shift operator $\Tjk$ influences the resulting subimage $\Imjk$. However, note that if $I$ is an OCT image of a single distant point-source, i.e., $I(x) = \delta(x)$ with $\delta$ denoting the delta-distribution, then $\F(I) \equiv  1$. Hence, in this case the shift operator $\Tjk$ is not relevant, and thus all subimages $\Imjk$ are low-resolution versions of the image $I$. Furthermore, as we shall see below, in case of aberrations, i.e., if $\phi(\xiv) \not\equiv 0$, these subimages are shifted depending mainly on the average slope of the wavefront aberration $\phi$ on the subaperture, i.e., of $\phijk$. This fact is leveraged by the subaperture-correlation method.

% % % % % % % % % % % % % % % % % % % % % % % % % %
% Subsection - The Subaperture-Correlation Method %
% % % % % % % % % % % % % % % % % % % % % % % % % %
\subsection{The Subaperture-Correlation Method}

In this section, we revisit the subaperture-correlation method of \cite{Kumar_Drexler_Leitgeb_2013,Kumar_Wurster_Salas_Ginner_Drexler_Leitgeb_2017}. The physical assumptions underlying this method translate to the following mathematical

\begin{assumption}\label{ass_corr}
	On each subaperture $\Ojk$, the wavefront aberration $\phi$ can be reasonably well approximated by a linear function, i.e., for all $j,k \in \Z$ there holds
	\begin{equation}\label{ass_phijk_linear}
		\phijk(\xiv) =  2\pi \kl{ \cjk +  \svjk \cdot \xiv + \etajk(\xiv)} \,,
		\qquad
		\forall \, \xiv \in \Omega_{0,0} \,,
	\end{equation}
    where the functions $\etajk$ are small perturbations. Furthermore, the Fourier information is \emph{self-similar}, i.e., for all $j,k \in \Z$ and with small perturbation functions $\dejk$ there holds
	\begin{equation}\label{ass_F_selfsimilar}
		\F(I)(\Tjk(\xiv)) = \F(I)(\xiv) + \dejk(\xiv)\,,
		\qquad
		\forall \, \xiv \in \Omega_{0,0} \,.
	\end{equation}
\end{assumption}
Using these assumptions, we are able to derive the following

\begin{proposition}\label{prop_subimages_old}
Let $\Imjk$ be defined as in \eqref{def_Imjk} and let Assumption~\ref{ass_corr} hold. Then 
	\begin{equation}\label{subimages_shifted}
		\abs{ \Imjk(\xv - \svjk)} 
        = 
        \abs{ \Imzz(\xv - \svzz)} 
        + \Ejk(\xv)\,,
		\qquad
		\forall \, j,k\in\Z \,,
	\end{equation}
where the perturbation functions $\Ejk$ are given by
    \begin{equation}\label{eq_Dejk}
    \begin{split}
        \Ejk(\xv) &:= 
        \abs{\FI\kl{\chi_{0,0}(\xiv) e^{2\pi i \etajk(\xiv)} \kl{ \F\kl{I}(\xiv) + \dejk(\xiv)}}(\xv)}
        \\
        & \qquad -
        \abs{\FI\kl{\chi_{0,0}(\xiv) e^{2\pi i \etazz(\xiv)} \kl{ \F\kl{I}(\xiv) + \dezz(\xiv)}}(\xv)} \,.
    \end{split}
    \end{equation}
Furthermore, the following error estimate holds:
    \begin{equation}\label{est_Dejk}
    \begin{split}
        & \norm{ \abs{ \Imjk(\cdot - \svjk)} - \abs{\Imzz(\cdot - \svzz)}  }_{\LiRt} 
        = \norm{\Ejk}_{\LiRt}
        \\
        & \qquad
        \leq
        2 \pi \norm{\F(I) }_{L^1(\Ozz)} \norm{\etajk - \etazz}_{L^\infty(\Ozz)}
        +
        \norm{\dejk }_{L^1(\Ozz)} + \norm{\dezz }_{L^1(\Ozz)}  \,.
    \end{split}
    \end{equation}
\end{proposition}
\begin{proof}
	First, note that by using \eqref{ass_phijk_linear} in \eqref{Imjk_exp} it follows that
	\begin{equation*}
		\begin{split}
			\Imjk(\xv) 
			&=
			\FI\kl{\chi_{0,0}(\xiv) e^{i \phijk(\xiv)} \F(I)(\Tjk(\xiv)) }(\xv)
			\\
			& = 
			\FI\kl{\chi_{0,0}(\xiv) e^{2\pi i \kl{ \cjk + \svjk \cdot \xiv + \etajk(\xiv) }} \F(I)(\Tjk(\xiv)) }(\xv) \,,
		\end{split}
	\end{equation*}
	which together with \eqref{prop_modulation} implies that
	\begin{equation*}
		\begin{split}
			\Imjk(\xv)  
			& = 
			\FI\kl{\chi_{0,0}(\xiv) e^{2\pi i \kl{ \cjk + \svjk \cdot \xiv + \etajk(\xiv) }}\F(I)(\Tjk(\xiv))}(\xv) \,,
			\\
			& = e^{2\pi i \cjk  } \FI\kl{\chi_{0,0}(\xiv) e^{2\pi i \etajk(\xiv)}  \F(I)(\Tjk(\xiv))}(\xv + \svjk) \,.
		\end{split}
	\end{equation*}
	Together with \eqref{ass_F_selfsimilar} it follows that 
	\begin{equation*}
		\Imjk(\xv) 
        =
        e^{2\pi i \cjk  } \FI\kl{\chi_{0,0}(\xiv) e^{2\pi i \etajk(\xiv)} \kl{ \F\kl{I}(\xiv) + \dejk(\xiv)}}(\xv + \svjk) \,,
	\end{equation*}
	which can be rewritten into
	\begin{equation*}
		e^{- 2\pi i \cjk  } \Imjk(\xv - \svjk) =
        \FI\kl{\chi_{0,0}(\xiv) e^{2\pi i \etajk(\xiv)} \kl{ \F\kl{I}(\xiv) + \dejk(\xiv)}}(\xv)\,,
	\end{equation*}
    and which after taking the absolute value on both sides yields
    \begin{equation*}
		\abs{ \Imjk(\xv - \svjk)} 
        =
        \abs{\FI\kl{\chi_{0,0}(\xiv) e^{2\pi i \etajk(\xiv)} \kl{ \F\kl{I}(\xiv) + \dejk(\xiv)}}(\xv)}\,.
	\end{equation*}
 	Since this holds in particular for $j = k = 0$, i.e.,
    \begin{equation*}
		\abs{ \Imzz(\xv - \svzz)} 
        =
        \abs{\FI\kl{\chi_{0,0}(\xiv) e^{2\pi i \etazz(\xiv)} \kl{ \F\kl{I}(\xiv) + \dezz(\xiv)}}(\xv)}\,,
	\end{equation*}
    it follows that \eqref{subimages_shifted} holds with $\Ejk$ as in \eqref{eq_Dejk}. Concerning the estimate \eqref{est_Dejk}, note first that due to the reverse triangle inequality $\abs{\abs{a}-\abs{b}} \leq \abs{a-b}$ it follows from \eqref{eq_Dejk} that
        \begin{equation*}
        \begin{split}
            \abs{\Ejk(\xv)}
            &\leq
            \big\vert \FI\kl{\chi_{0,0}(\xiv) e^{2\pi i \etajk(\xiv)} \kl{ \F\kl{I}(\xiv) + \dejk(\xiv)}}(\xv)
            \\
            & \qquad \qquad \qquad -
            \FI\kl{\chi_{0,0}(\xiv) e^{2\pi i \etazz(\xiv)} \kl{ \F\kl{I}(\xiv) + \dezz(\xiv)}}(\xv)
            \big\vert   \,.
        \end{split}
        \end{equation*}
    which after rearranging the terms yields
        \begin{equation*}
        \begin{split}
            \abs{\Ejk(\xv)}
            &\leq
            \big\vert \FI\kl{\chi_{0,0}(\xiv) \F\kl{I} \kl{ e^{2\pi i \etajk(\xiv)} - e^{2\pi i \etazz(\xiv)} }}
            \\
            \qquad &+            
            \FI\kl{\chi_{0,0} \kl{ \dejk(\xiv)e^{2\pi i \etajk(\xiv)} - \dezz(\xiv) e^{2\pi i \etazz(\xiv)}  }}
            \big\vert   \,.
        \end{split}
        \end{equation*}
    Now using the definition of the Fourier transform and the triangle inequality we obtain
        \begin{equation*}
        \begin{split}
            \abs{\Ejk(\xv)}
            &\leq
            \int_{\Ozz} \abs{ \F\kl{I}(\xiv) \kl{ e^{2\pi i \etajk(\xiv)} - e^{2\pi i \etazz(\xiv)} }} \, d \xiv 
            \\
            \qquad &+            
            \int_{\Ozz} \abs{ \dejk(\xiv)e^{2\pi i \etajk(\xiv)} - \dezz(\xiv) e^{2\pi i \etazz(\xiv)}  } \, d \xiv
            \,,
        \end{split}
        \end{equation*}
    which can further be estimated by
        \begin{equation*}
        \begin{split}
            \abs{\Ejk(\xv)}
            \leq
            \norm{\F(I) }_{L^1(\Ozz)} \max_{\xiv \in \Ozz}  \abs{ e^{2\pi i \etajk(\xiv)} - e^{2\pi i \etazz(\xiv)} } 
            +            
            \norm{\dejk }_{L^1(\Ozz)} + \norm{\dezz }_{L^1(\Ozz)} 
            \,.
        \end{split}
        \end{equation*}
    Together with the fact that due to the inequality $\abs{e^{2 \pi i a} - 1} \leq 2 \pi \abs{a}$ there holds
        \begin{equation*}
            \abs{ e^{2\pi i \etajk(\xiv)} - e^{2\pi i \etazz(\xiv)} }
            =
            \abs{ e^{2\pi i (\etajk(\xiv)- \etazz(\xiv))} - 1 }
            \leq 
            2 \pi \abs{\etajk(\xiv)- \etazz(\xiv)} \,,
        \end{equation*}
    it thus follows that 
        \begin{equation*}
        \begin{split}
            \abs{\Ejk(\xv)}
            \leq
            2 \pi \norm{\F(I) }_{L^1(\Ozz)} \norm{\etajk - \etazz}_{L^\infty(\Ozz)}
            +          
            \norm{\dejk }_{L^1(\Ozz)} + \norm{\dezz }_{L^1(\Ozz)}
            \,,
        \end{split}
        \end{equation*}    
    which yields \eqref{est_Dejk} and thus completes the proof.
\end{proof}

Note that \eqref{est_Dejk} is only a rough estimate for the error terms $\Ejk$. Nevertheless, it implies that if the perturbations $\etajk$ and $\dejk$ are small, it follows that
    \begin{equation}\label{subimages_shifted_approx}
		\abs{ \Imjk(\xv - \svjk)} 
        \approx 
        \abs{ \Imzz(\xv - \svzz)} \,,
		\qquad
		\forall \, j,k\in\Z \,,
	\end{equation}
which in \cite{Kumar_Drexler_Leitgeb_2013,Kumar_Wurster_Salas_Ginner_Drexler_Leitgeb_2017} is used to estimate the relative shifts $(\svjk-\svzz)$ by matching the subimages $\Imjk$ and $\Imzz$. From these shifts, the full wavefront aberration function $\phi$ is reconstructed via a least squares fitting onto either a piecewise linear or a Zernike polynomial basis. The subaperture-corellation method can thus be summarized into the following:

\begin{algo}[Subaperture-Correlation Method]\label{algo_old}
	\hspace{0pt}
	\begin{enumerate}
		\item For given $\Im$ and each $j,k \in \Z$ compute the subimages $\Imjk$ defined in \eqref{def_Imjk}.
		\item Estimate the relative slopes $(\svjk-\svzz)$ through image fitting using \eqref{subimages_shifted_approx}.
		\item Reconstruct the wavefront aberration $\phi$ from the relative slopes $(\svjk-\svzz)$.
		\item Compute the corrected image by applying the phase correction \eqref{image_correction}.
	\end{enumerate}
\end{algo}

Concerning the second step in the above algorithm, note that instead of via image fitting, the relative slopes $(\svjk-\svzz)$ can also be computed explicitly, as we see in

\begin{proposition}
	Let $\Imjk$ be defined as in \eqref{def_Imjk} and let Assumption~\ref{ass_corr} hold with $\etajk = 0$ and $\dejk = 0$ for all $j,k \in \Z$. Then
	\begin{equation}\label{subimages_slopes}
		\svjk - \svzz
		\, = \,
		\frac{\int_\Rt \xv \abs{\Imzz(\xv) } \, d\xv}{\int_\Rt \abs{\Imzz(\xv) } \, d\xv}
		-
		\frac{\int_\Rt \xv \abs{  \Imjk(\xv ) } \, d\xv}{\int_\Rt \abs{\Imjk(\xv) } \, d\xv} \,,
		\qquad
		\forall \, j,k\in\Z \,.
	\end{equation}
\end{proposition}
\begin{proof}
	Let $j,k \in \Z$ be arbitrary but fixed. Due to \eqref{subimages_shifted} and $\etajk = \dejk = 0$ there holds
	\begin{equation*}
		\abs{\Imjk(\xv - (\svjk - \svzz)) }
		=
		\abs{  \Imzz(\xv ) } \,.
	\end{equation*}
	Computing the center of mass on both sides we obtain
	\begin{equation}\label{eq_helper_03}
		\frac{\int_\Rt \xv \abs{\Imjk(\xv - (\svjk - \svzz)) } \, d\xv}{\int_\Rt \abs{\Imjk(\xv - (\svjk - \svzz)) } \, d\xv}
		= 
		\frac{\int_\Rt \xv \abs{  \Imzz(\xv ) } \, d\xv}{\int_\Rt \abs{  \Imzz(\xv ) } \, d\xv}\,.
	\end{equation}
	Next, note that by a simple change of variables in the integral we obtain
	\begin{equation*}
		\int_\Rt \xv \abs{\Imjk(\xv - (\svjk-\svzz)) } \, d\xv
		=
		\int_\Rt \xv \abs{\Imjk(\xv) } \, d\xv
		+
		(\svjk-\svzz) \int_\Rt \abs{\Imjk(\xv) } \, d\xv \,.
	\end{equation*}
	Together with the fact that 
	\begin{equation*}
		\int_\Rt \abs{\Imjk(\xv - (\svjk-\svzz)) } \, d\xv = 
		\int_\Rt \abs{\Imjk(\xv) \, } \, d\xv \,,
	\end{equation*}
	we obtain that \eqref{eq_helper_03} is equivalent to
	\begin{equation*}
		\frac{\int_\Rt \xv \abs{\Imjk(\xv) } \, d\xv
			+ (\svjk-\svzz) \int_\Rt \abs{\Imjk(\xv) } \, d\xv}{\int_\Rt \abs{\Imjk(\xv) \, } \, d\xv }
		= 
		\frac{\int_\Rt \xv \abs{  \Imzz(\xv ) } \, d\xv}{\int_\Rt \abs{  \Imzz(\xv ) } \, d\xv}\,.
	\end{equation*}
	which after a reordering of the terms now yields the assertion.
\end{proof}

Note that for the third step in Algorithm~\ref{algo_old}, any wavefront reconstruction method for the SH-WFS such as the CuReD method introduced in Section~\ref{subsect_SH_WFS} can be used. Also, note that condition \eqref{ass_F_selfsimilar} is satisfied if $I$ is the OCT image of a distant point source, since then $\F(I) \equiv  1$. Furthermore, it was found experimentally in  \cite{Kumar_Drexler_Leitgeb_2013,Kumar_Wurster_Salas_Ginner_Drexler_Leitgeb_2017} that condition \eqref{ass_F_selfsimilar} is essential for the subaperture-correlation method to work satisfactorily. However, while in many cases the Fourier information is self-similar and thus \eqref{ass_F_selfsimilar} is satisfied, this ultimately depends on the structure of the considered objects, thus limiting its applicability. Hence, it is desirable to find a DAC approach which does not require a condition like \eqref{ass_F_selfsimilar}. However, we already indicated that a DAC approach requires some type of assumption. Thus, in the next section we will only use a minimal additional assumption, which has a statistical justification and leads to a novel subaperture-based DAC approach.

% % % % % % % % % % % % % % % % % % % % %
% Section - A Novel Correction Approach %
% % % % % % % % % % % % % % % % % % % % % 
\section{A Novel Correction Approach}\label{sect_novel_DAC}

In this section, we present a novel aberration correction approach for OCT based on virtual subapertures, a statistical distribution assumption, and a suitable averaging.

% % % % % % % % % % % % % % % % % % % % %
% Subsection - Motivation and Key Ideas %
% % % % % % % % % % % % % % % % % % % % %
\subsection{Motivation and Key Ideas}\label{subsect_novel_DAC_motivation}

The following approach is motivated by the observation that in hardware-based AO-OCT as used, e.g., in ophthalmology, the wavefront sensing is typically complemented with an averaging process to remove the speckle noise due to the object structure \cite{Hofer_Artal_Singer_Aragon_Williams_2001}. The averaging blurs out the object structure phase of the OCT image and leaves only the wavefront error to be detected, resulting in more well-defined focal spots on the SH-WFS and more accurate wavefront aberration measurements. Mathematically, this indicates that the spectral phase and the wavefront phase are generally of a different frequency, and can thus be distinguished from each other.

% % % % % % % % % % % % % % % % % % % %
% Figure - Subdivision of Subdomains  %
% % % % % % % % % % % % % % % % % % % %
\begin{figure}[ht!]
	\centering
	\includegraphics[width=\textwidth]{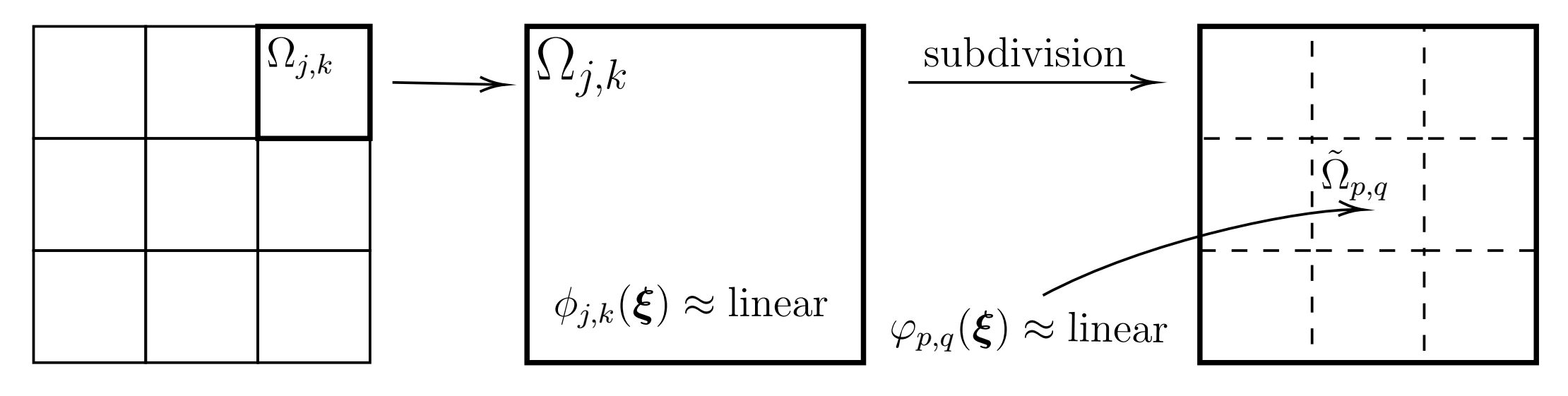}
	\caption{Illustration of the mathematical setting used in our novel DAC method: Subdivision of $\R^2$ into subdomains $\Ojk$ (left), single subdomain $\Ojk$ and assumption that $\phijk$ is approximately linear on $\Ojk$ (middle), subdivision of $\Ojk$ into smaller subdomains $\Otpq$ and assumption that $\vphipq$ is approximately linear on $\Otpq$ (right).} 
	\label{fig_subdivision_subdomains}
\end{figure}

In order to use these observations for deriving a novel DAC approach, we now consider the Fourier transform $\F(I)$ and note that it can be written in the polar form
\begin{equation}\label{eq_F_I_polar}
	\F(I)(\xiv) = \abs{\F(I)(\xiv)}e^{i\vphi(\xiv)} \,,
\end{equation}
where $\abs{\F(I)(\xiv)}$ is called the \emph{spectral amplitude} of $I$, and the real valued function $\vphi(\xiv)$ denotes the angular component of $\F(I)$, and is referred to as the \emph{spectral phase}.

In order to translate the physical observations made above into a mathematically useful form, we return again to the subapertures $\Ojk$. As before, we will assume that the wavefront aberration $\phi$ can be reasonably well approximated by linear functions $\phijk$ on the subdomains $\Ojk$. Next, we translate the physical observation that the spectral phase $\vphi$ introduced by the object contains higher frequencies than the phase $\phi$ corresponding to the wavefront aberrations. For this, we subdivide each subaperture $\Ojk$ into a number of smaller subapertures $\Otpq$. While the precise definition is given below, for now we refer to the illustration given in Figure~\ref{fig_subdivision_subdomains}. On each of the smaller subdomains $\Otpq$ we now approximate the spectral phase $\vphi$ by a linear function $\vphipq$. Then, the physical observation that the highly oscillatory spectral phase $\vphi$ causes generally random phase shifts translates to the mathematical statement that the slopes of $\vphipq$ are randomly distributed with a zero mean. We now proceed to make these statements mathematically precise.

% % % % % % % % % % % % % % % % % % % % % % % % % % % % % % % % %
% Subsection - Subdomains/Subapertures and Subimages Revisited  %
% % % % % % % % % % % % % % % % % % % % % % % % % % % % % % % % %
\subsection{Subdomains/Subapertures and Subimages Revisited}\label{subsect_sub_02}

In order to mathematically specify the physical assumptions made above, we need to adapt the definitions of subdomains/subapertures and the resulting subimages made in Section~\ref{subsect_sub_01}. We start with the definition of the subdivisions of the subdomains $\Ojk$.

\begin{definition}
	For $P, Q \in \R^+$ we define the subdomain/subaperture
	\begin{equation*}
		\Otzz := [ -P/2, P/2)  \times [-Q/2, Q/2)  \subset \R^2 \,,
	\end{equation*}
	as well as the shifted subdomains/subapertures
	\begin{equation*}
		\Otpq := \Ttpq( \Otzz ) 
		= 
		\Kl{\xiv \in \R^2 \, \vert \, \Ttpq^{-1}(\xiv) \in \Otzz} \,,
		\qquad
		\forall \, p,q\in \Z \,,
	\end{equation*}
	where the translation operator $\Ttpq$ is defined by
	\begin{equation}\label{def_Ttpq}
		\Ttpq(\xiv) := (\xi_1 + pP, \xi_2 + qQ) \,,
		\qquad
		\forall \, \xiv \in \R^2 \,,
		\quad
		\forall \, p,q \in \Z \,. 
	\end{equation}
\end{definition}

Throughout this manuscript, we assume that $J, K$ are certain multiples of $P, Q$, i.e.,
\begin{equation}\label{ass_mult}
	\exists \, l_1,l_2 \in \N \, : \quad J = (2 l_1 + 1) P \,,
	\quad \text{and} \quad K = (2 l_2 + 1) Q \,. 
\end{equation}
This implies that every subdomain $\Ojk$ can be uniquely decomposed into a finite union of the subdomains $\Otpq$; see Figure~\ref{fig_subdivision_subdomains}. More precisely, if we define the index sets
\begin{equation*}
	\Djk := \Kl{ (p,q) \in \Z^2 \, \vert \, \Otpq \subseteq \Ojk   } \,,
\end{equation*}
then \eqref{ass_mult} implies that
\begin{equation*}
	\Ojk = \bigcup_{p,q \in \Djk} \Otpq \,.
\end{equation*}
Now as before, for all $p,q\in\Z$ and $\xiv \in \R^2$ we consider the indicator functions
\begin{equation*}
	\chitpq(\xiv) 
	:=
	\begin{cases}
		1 \,, & \xiv \in \Otpq \,,
		\\
		0 \,, & \text{else} \,.
	\end{cases}
\end{equation*}
of the subdomains $\Otpq$, and define the following ``cut-outs'' of the Fourier coefficients:
\begin{equation}\label{def_Dmtpq}
	\Dtmpq(\xiv) :=  \chitpq(\xiv) D^m(\xiv)  = \chitpq(\xiv) e^{i \phi(\xiv)} \F(I)(\xiv) \,,
	\qquad
	\forall \, \xiv \in \R^2 \,.
\end{equation}
The subimages corresponding to those Fourier ``cut-outs'' are given in the following

\begin{definition}
	For each $p,q \in \Z$, the subimages $\Itmpq$ are defined via
	\begin{equation}\label{def_Imtpq}
		\Itmpq(\xv) := \FI( \Dtmpq( \Ttpq(\xiv)))(\xv)\,,
		\qquad
		\forall \, \xv \in \R^2 \,,
	\end{equation}
	where the cut-out Fourier data $\Dtmpq$ are given by \eqref{def_Dmtpq}. \end{definition}

Similarly to Proposition~\ref{prop_subimages}, we can derive an alternative representation in

\begin{proposition}
	Let the subimages $\Itmpq$ be as defined in \eqref{def_Imtpq} and let 
	\begin{equation}\label{def_phitpq}
		\phitpq(\xiv) := \phi(\Ttpq(\xiv)) \,,
		\qquad
		\forall \, \xiv \in \Otzz \,,
		\quad
		\forall \, p,q \in \Z \,,
	\end{equation}
	denote the local wavefront aberrations. Then we obtain the representation
	\begin{equation}\label{Imtpq_exp}
		\Itmpq(\xv) = \FI\kl{\chitzz(\xiv) e^{i \phitpq(\xiv)} \F(I)(\Ttpq(\xiv))}(\xv) \,.
	\end{equation}
\end{proposition}
\begin{proof}
	This follows in direct analogy to the proof of Proposition~\ref{prop_subimages}.
\end{proof}

% % % % % % % % % % % % % % % % % % % % % % % % % %
% Section - Mathematical Derivation of the Method %
% % % % % % % % % % % % % % % % % % % % % % % % % %
\subsection{Mathematical Derivation of the Method}

The physical motivation and key ideas introduced in Section~\ref{subsect_novel_DAC_motivation} now translate to 

\begin{assumption}\label{ass_new}
	On each subaperture $\Ojk$, the wavefront aberration $\phi$ can be reasonably well approximated by a linear function, i.e., \eqref{ass_phijk_linear} holds for all $j,k \in \Z$, i.e.,
	\begin{equation*}
		\phijk(\xiv) = 2\pi \kl{ \cjk +  \svjk \cdot \xiv + \etajk(\xiv)} \,,
		\qquad
		\forall \, \xiv \in \Omega_{0,0} \,.
	\end{equation*}
	Furthermore, on each smaller subaperture $\Otpq$, the spectral phase $\vphi$ can also be reasonably well approximated by a linear function, i.e., for all $p,q \in \Z$ there holds
	\begin{equation}\label{ass_vphipq_linear}
		\vphipq(\xiv) := \vphi(\Ttpq(\xiv)) =
        2\pi \kl{ \dpq +  \tvpq \cdot \xiv + \zetapq(\xiv)}  \,,
		\qquad
		\forall \, \xiv \in \Otzz \,.
	\end{equation} 
	Moreover, the magnitude of the Fourier information is \emph{self-similar}, i.e., for all $p,q \in \Z$,
	\begin{equation}\label{ass_F_abs_selfsimilar}
		\abs{\F(I)(\Ttpq(\xiv))} =
        \abs{\F(I)(\xiv)} + \gapq(\xiv) \,,
		\qquad
		\forall \, \xiv \in \Otzz \,.
	\end{equation}  
	Finally, let $J,K,P,Q \in \R^+$ satisfy the multiplicity condition \eqref{ass_mult}, and let
	\begin{equation}\label{eq_spectral_phase_mean}
		\sum_{p,q \in \Djk} \kl{\tvpq - \tvzz} =\epsjk \,,
		\qquad
		\forall \, j,k \in \Z \,.
	\end{equation}
    Here, the functions $\etajk$, $\zetapq$, and $\gapq$, and the constants $\epsjk$ are small perturbations.
\end{assumption}

Note first that \eqref{eq_spectral_phase_mean} encodes the physical observation that averaging blurs out the spectral phase of the OCT image: as we will see below, the relative slopes $(\tvpq - \tvzz)$ cause local shifts in the subimages $\Itmpq$ on a smaller scale than those induced by the relative slopes $(\svjk - \svzz)$. The fact that in hardware-based AO with a SH-WFS these local shifts can be removed by averaging thus motivates \eqref{eq_spectral_phase_mean}. Recall that without including such additional physical information, the general DAC Problem~\ref{problem_DAC} is not uniquely solvable, since essentially the wavefront aberration $\phi$ has to be determined from the spectral phase $\phi + \vphi$ of the measured image $\Im$, with $\vphi$ also being unknown in general. Assumption~4.4 thus imposes additional, physically motivated restrictions, which can informally be understood to lead to a solvable problems as follows: On each large subaperture $\Ojk$ we are interested in one variable, namely the average slope $\svjk$. Each large subaperture $\Ojk$ is in turn divided into $((2 l_1 + 1) \cdot (2 l_2 + 1))$ smaller subapertures $\Otpq$, on which also the average slopes $\tvpq$ are unknown. Hence, for each large subaperture we have $((2 l_1 + 1) \cdot (2 l_2 + 1) + 1)$ unknowns, while we only have access to $((2 l_1 + 1) \cdot (2 l_2 + 1))$ ``equations'', namely the spectral phase information $\arg(\F(\Im))$. Hence, an additional restriction is required, which is exactly the physically motivated averaging assumption \eqref{eq_spectral_phase_mean}, which is thus the key to our approach.

Finally, note that in contrast to \eqref{ass_F_selfsimilar}, we now only require that the spectral amplitude $\abs{\F(I)}$ is self-similar. Since due to \eqref{model} the wavefront aberration $\phi$ acts only on the spectral phase $\vphi$ but not on the spectral amplitude $\abs{\F(I)}$, this assumption can always be circumvented in practice, as outlined in Remark~\ref{remark_selfsim} below.

For the upcoming derivation, we first need the following technical result:
\begin{proposition}\label{prop_technical}
	Let Assumption~\ref{ass_new} hold, let $j,k\in\Z$ and let $p,q \in \Djk$. Then with $\ctjkpq := \cjk + \svjk \cdot (pP-jJ , q Q - kK)$ and $\etatjkpq(\xiv) := \etajk(\xi_1+pP-jJ , \xi_2 + q Q - kK)$ the function $\phitpq$ defined in \eqref{def_phitpq} satisfies
	\begin{equation*}
		\phitpq(\xiv) 
		=
        2\pi \kl{ \ctjkpq +	\svjk \cdot \xiv
        + \etatjkpq(\xiv)} \,,
		\qquad 
		\forall \, \xiv \in \Otzz \,,
	\end{equation*}
\end{proposition}
\begin{proof}
	Let $j,k \in \Z$, $p,q \in \Djk$, and $\xiv \in \Otzz$, be arbitrary but fixed and consider
	\begin{equation*}
		\phitpq(\xiv) \overset{\eqref{def_phitpq}}{=} \phi(\Ttpq(\xiv)) 
		\overset{\eqref{def_Ttpq}}{=} \phi(\xi_1+pP , \xi_2 + q Q) \,.
	\end{equation*}
	Since due to \eqref{def_Tjk} and \eqref{def_phijk} there holds
	\begin{equation*}
		\begin{split}
			\phi(\xi_1+pP , \xi_2 + q Q)
			&= \phi((\xi_1+pP-jJ) + jJ , (\xi_2 + q Q - kK) + kK)
			\\
			&\overset{\eqref{def_Tjk}}{=} \phi(\Tjk(\xi_1+pP-jJ , \xi_2 + q Q - kK))
			\\
			&\overset{\eqref{def_phijk}}{=} \phijk(\xi_1+pP-jJ , \xi_2 + q Q - kK) \,,
		\end{split}    
	\end{equation*}
	it follows that
	\begin{equation*}
		\phitpq(\xiv) = \phijk(\xi_1+pP-jJ , \xi_2 + q Q - kK) \,.
	\end{equation*}
	Now since $\xiv \in \Otzz$ and $p,q \in \Djk$ it follows that $(\xi_1+pP-jJ , \xi_2 + q Q - kK) \in \Omega_{0,0}$. Hence, it follows with \eqref{ass_phijk_linear} that
	\begin{equation*}
		\begin{split}
			\phitpq(\xiv) 
			&= \phijk(\xi_1+pP-jJ , \xi_2 + q Q - kK)
			\\
			& = 2\pi \kl{ \cjk +  \svjk \cdot (\xi_1+pP-jJ , \xi_2 + q Q - kK) + \etajk(\xi_1+pP-jJ , \xi_2 + q Q - kK) } 
			\\
			& = 2\pi \kl{ \cjk + \svjk \cdot (pP-jJ , q Q - kK)
			+
			\svjk \cdot \xiv
            + \etajk(\xi_1+pP-jJ , \xi_2 + q Q - kK)} \,,
		\end{split}    
	\end{equation*}.
	which together with the definitions of $\ctjkpq$ and $\etatjkpq$ now yields the assertion.
\end{proof}

Using the above proposition, we can now proceed to the following

\begin{theorem}
	Let $\Itmpq$ be defined as in \eqref{def_Imtpq} and let \eqref{model} and Assumption~\ref{ass_new} hold. Then for each $j,k \in \Z$, any $p,q \in \Djk$, and for all $\xv \in \R^2$ there holds
	\begin{equation}\label{eq_shift}
		\abs{\Itmpq(\xv - \kl{\svjk + \tvpq})} = \abs{\Itmzz(\xv - \kl{\svzz + \tvzz})} + \Etpq(\xv)  \,,
	\end{equation}
where the perturbation functions $\Etpq$ are given by
    \begin{equation}\label{eq_Detpq}
    \begin{split}
        \Etpq(\xv) &:= 
        \abs{ \FI\kl{\chitzz(\xiv) e^{2\pi i \kl{\etatjkpq(\xiv) + \zetapq(\xiv)}} \kl{ \abs{\F\kl{I}(\xiv) }+ \gapq(\xiv)}}(\xv)}
        \\
        & \qquad -
        \abs{ \FI\kl{\chitzz(\xiv) e^{2\pi i \kl{\etatzzzz(\xiv) + \zetazz(\xiv)}} \kl{ \abs{\F\kl{I}(\xiv)} + \gazz(\xiv)}}(\xv)} \,.
    \end{split}
    \end{equation}
Furthermore, the following error estimate holds:
    \begin{equation}\label{est_Detpq}
    \begin{split}
        & \norm{ \abs{\Itmpq(\cdot - \kl{\svjk + \tvpq})}   -  \abs{\Itmzz(\cdot - \kl{\svzz + \tvzz})} }_{\LiRt}
        =
        \norm{ \Etpq }_{\LiRt}
        \\
        &\qquad \leq
        2 \pi \norm{\F(I) }_{L^1(\Otzz)} \kl{ \norm{\etajk}_{L^\infty(\Ozz)}
        + \norm{\etazz}_{L^\infty(\Ozz)} + \norm{\zetapq-\zetazz}_{L^\infty(\Otzz)}  }
        \\
        & \qquad \qquad +          
        \norm{\gapq }_{L^1(\Otzz)} + \norm{\gazz }_{L^1(\Otzz)}\,.
    \end{split}
    \end{equation}
\end{theorem}
\begin{proof}
	Let $j,k \in \Z$ be arbitrary but fixed and let $p,q\in\Djk$. Due to \eqref{Imtpq_exp} there holds
	\begin{equation*}
		\Itmpq(\xv) = \FI\kl{\chitzz(\xiv) e^{i \phitpq(\xiv)} \F(I)(\Ttpq(\xiv))}(\xv) \,.
	\end{equation*}
	Since due to Proposition~\ref{prop_technical} there holds
	\begin{equation*}
		\phitpq(\xiv) 
		=
        2\pi \kl{ \ctjkpq +	\svjk \cdot \xiv + \etatjkpq(\xiv)} \,,
		\qquad 
		\forall \, \xiv \in \Otzz \,,
	\end{equation*}
	it follows that
	\begin{equation*}
		\Itmpq(\xv) 
        =
        \FI\kl{\chitzz(\xiv) e^{2\pi i \kl{ \ctjkpq + \svjk \cdot \xiv + \etatjkpq(\xiv)} } \F(I)(\Ttpq(\xiv))}(\xv) \,,
	\end{equation*}
	which can be rearranged into
	\begin{equation}\label{eq_helper_01}
		\Itmpq(\xv) 
        =
        e^{2\pi i \ctjkpq } \FI\kl{\chitzz(\xiv) e^{2\pi i (\svjk \cdot \xiv + \etatjkpq(\xiv))}  \F(I)(\Ttpq(\xiv))}(\xv) \,.
	\end{equation}
	Next, note that due to \eqref{eq_F_I_polar} and \eqref{ass_vphipq_linear} for all $\xiv \in \Otzz$ there holds  
	\begin{equation*}
		\begin{split}
			\F\kl{I}(\Ttpq(\xiv))
			& \overset{\eqref{eq_F_I_polar}}{=}
			\abs{\F\kl{I}(\Ttpq(\xiv))} e^{i \vphi(\Ttpq(\xiv)))}
			\overset{\eqref{ass_vphipq_linear}}{=}
			\abs{\F\kl{I}(\Ttpq(\xiv))} e^{i 2\pi \kl{ \dpq +  \tvpq \cdot \xiv + \zetapq(\xiv)}}
		\end{split}
	\end{equation*}    
	Using this in \eqref{eq_helper_01} and after rearranging we obtain    
	\begin{equation*}
		\Itmpq(\xv) = e^{2\pi i \kl{\ctjkpq + \dpq}} \FI\kl{\chitzz(\xiv) e^{2\pi i \kl{\kl{\svjk + \tvpq} \cdot \xiv + \etatjkpq(\xiv) + \zetapq(\xiv)}}  \abs{\F\kl{I}(\Ttpq(\xiv))}}(\xv) \,,
	\end{equation*}
	which together with \eqref{prop_modulation} and after taking the absolute value yields    
	\begin{equation}\label{eq_helper_02}
		\abs{\Itmpq(\xv)} = \abs{ \FI\kl{\chitzz(\xiv) e^{2\pi i \kl{\etatjkpq(\xiv) + \zetapq(\xiv)}} \abs{\F\kl{I}(\Ttpq(\xiv))}}(\xv + \kl{\svjk + \tvpq})} \,.
	\end{equation}
	Next, note that due to the self-similarity assumption \eqref{ass_F_selfsimilar} it follows that 
	\begin{equation*}
		\abs{\Itmpq(\xv)} = \abs{ \FI\kl{\chitzz(\xiv) e^{2\pi i \kl{\etatjkpq(\xiv) + \zetapq(\xiv)}} \kl{ \abs{\F\kl{I}(\xiv)} + \gapq(\xiv)}}(\xv + \kl{\svjk + \tvpq})} \,,
	\end{equation*}
	and thus after the change of variables $x \rightarrow x - \kl{\svjk + \tvpq}$ we obtain
	\begin{equation*}
		\abs{\Itmpq(\xv - \kl{\svjk + \tvpq})} = \abs{ \FI\kl{\chitzz(\xiv) e^{2\pi i \kl{\etatjkpq(\xiv) + \zetapq(\xiv)}} \kl{ \abs{\F\kl{I}(\xiv)} + \gapq(\xiv)}}(\xv)} \,.
	\end{equation*}
    Now, since $j,k \in \Z$ and $p,q \in \Djk$ were arbitrary, the above equation in particular also holds for $j=k=p=q=0$, i.e., 
    \begin{equation*}
		\abs{\Itmzz(\xv - \kl{\svzz + \tvzz})} = \abs{ \FI\kl{\chitzz(\xiv) e^{2\pi i \kl{\etatzzzz(\xiv) + \zetazz(\xiv)}} \kl{ \abs{\F\kl{I}(\xiv)} + \gazz(\xiv)}}(\xv)} \,,
	\end{equation*}
    and thus it follows that \eqref{eq_Detpq} holds with $\Etpq$ as in \eqref{eq_Detpq}. Concerning the estimate \eqref{eq_Detpq}, note first that following the same steps as in the proof of Proposition~\ref{prop_subimages_old} we obtain
        \begin{equation*}
        \begin{split}
            \abs{\Etpq(\xv)}
            &\leq
            2 \pi \norm{\F(I) }_{L^1(\Otzz)} \norm{(\etatjkpq - \etatzzzz) + (\zetapq-\zetazz)}_{L^\infty(\Otzz)}
            \\
            & \qquad +          
            \norm{\gapq }_{L^1(\Otzz)} + \norm{\gazz }_{L^1(\Otzz)}
            \,,
        \end{split}
        \end{equation*} 
    which together with the triangle inequality and the fact that
        \begin{equation*}
            \norm{(\etatjkpq - \etatzzzz)}_{L^\infty(\Otzz)}
            \leq
            \norm{\etajk}_{L^\infty(\Ozz)}
            + \norm{\etazz}_{L^\infty(\Ozz)}
        \end{equation*}
    yields \eqref{est_Dejk} and thus completes the proof.
\end{proof}

As in the subaperture-correlation method, due to \eqref{eq_shift} and \eqref{eq_Detpq} it now follows that if the perturbations $\etajk$, $\zetapq$ and $\gapq$ are small, then the error terms $\Etpq$ are small and thus
    \begin{equation}\label{eq_shift_approx}
		\abs{\Itmpq(\xv - \kl{\svjk + \tvpq})}  
        \approx \abs{\Itmzz(\xv - \kl{\svzz + \tvzz})}  \,,
        \qquad 
        \forall \, j,k\in\Z\,, \quad p,q \in \Djk \,.
	\end{equation}
Hence, in this case the relative shift vectors $\vvjkpq := \kl{\svjk - \svzz} + \kl{\tvpq - \tvzz}$ can now be estimated by matching the subimages $\Itmpq$ and $\Itmzz$. Summing these vectors $\vvjkpq$ over all $p,q \in \Djk$ for fixed $j,k \in \Z$ yields
\begin{equation*}
	\sum\limits_{p,q\in\Djk} \vvjkpq = \sum\limits_{p,q\in\Djk} \kl{\svjk - \svzz}  + \sum\limits_{p,q\in\Djk} \kl{\tvpq - \tvzz} \,.
\end{equation*}
Noting that $\svjk$ is independent of $p,q \in \Djk$ and using \eqref{eq_spectral_phase_mean} we thus obtain 
\begin{equation}\label{slopes_shifts}
	\svjk - \svzz = \frac{1}{\abs{\Djk}} \sum\limits_{p,q\in\Djk} \vvjkpq \,\,-\,\, \frac{\epsjk}{\abs{\Djk}}  \,,
\end{equation}
where $\abs{\Djk}$ denotes the number of elements in the set $\Djk$. These considerations now give rise to our novel DAC approach, which is summarized in the following 

\begin{algo}[Novel Digital Aberration Correction Method]\label{algo_novel}
	\hspace{0pt}
	\begin{enumerate}
		\item For given $\Im$ and each $p,q \in \Z$ compute the subimages $\Itmpq$ defined in \eqref{def_Imtpq}.
		\item Estimate the shift vectors $\vvjkpq = \kl{\svjk-\svzz} + \kl{\tvpq-\tvzz}$ through image fitting using \eqref{eq_shift_approx}.
		\item Compute the relative slopes $\svjk-\svzz$ via averaging as in formula \eqref{slopes_shifts}.
		\item Reconstruct the wavefront aberration $\phi$ from the relative slopes $\kl{\svjk - \svzz}$.
		\item Compute the corrected image by applying the phase correction \eqref{image_correction}.
	\end{enumerate}
\end{algo}

Concerning the second and third step in the above algorithm, note that similarly as in Proposition~\ref{prop_subimages}, the relative slopes $\kl{\svjk - \svzz}$ can also be computed explicitly, as we now show in

\begin{proposition}
	Let $\Itmpq$ be defined as in \eqref{def_Imtpq} and let Assumption~\ref{ass_new} hold with $\etajk = 0$, $\zetapq = 0$, and $\gapq = 0$. Then
	\begin{equation}\label{subimages_slopes_new}
		\svjk - \svzz 
		=
		\frac{1}{\abs{\Djk}} \sum\limits_{p,q \in \Djk} \kl{ 
			\frac{
				\int_\Rt \xv \abs{\Itmzz(\xv)}     \, d\xv}
			{\int_\Rt \abs{ \Itmzz(\xv)} \, d \xv} 
			-
			\frac{
				\int_\Rt \xv \abs{\Itmpq(\xv) } \, d\xv}
			{\int_\Rt \abs{\Itmpq(\xv) \, } \, d\xv}}
        - \frac{\epsjk}{\abs{\Djk}}
		\,,
	\end{equation}
	where $\abs{\Djk}$ denotes the number of elements in the set $\Djk$.
\end{proposition}
\begin{proof}
	Let $j,k \in \Z$ be arbitrary but fixed and let $p,q\in\Djk$. Due to \eqref{eq_shift} and the assumption that $\etajk = \zetapq = \gapq = 0$, there holds
	\begin{equation*}
		\abs{\Itmpq(\xv - \kl{\kl{\svjk - \svzz} + \kl{\tvpq - \tvzz}})} = \abs{\Itmzz(\xv)}  \,.
	\end{equation*}
	Computing the center of mass on both sides we obtain
	\begin{equation}\label{eq_helper_04}
		\frac{
			\int_\Rt \xv \abs{\Itmpq(\xv - \kl{\kl{\svjk - \svzz} + \kl{\tvpq - \tvzz}})}     \, d\xv}
		{\int_\Rt \abs{ \Itmpq(\xv - \kl{\kl{\svjk - \svzz} + \kl{\tvpq - \tvzz}})} \, d \xv} 
		=
		\frac{
			\int_\Rt \xv \abs{\Itmzz(\xv)}     \, d\xv}
		{\int_\Rt \abs{ \Itmzz(\xv)} \, d \xv} 
		\,. 
	\end{equation}
	Next, note that by a simple change of variables in the integral we obtain
	\begin{equation*}
		\begin{split}
			&\int_\Rt \xv \abs{\Itmpq(\xv - \kl{\kl{\svjk - \svzz} + \kl{\tvpq - \tvzz}})} \, d\xv
			\\
			& \qquad =
			\int_\Rt \xv \abs{\Itmpq(\xv) } \, d\xv
			+
			\kl{\kl{\svjk - \svzz} + \kl{\tvpq - \tvzz}} \int_\Rt \abs{\Itmpq(\xv) } \, d\xv \,.
		\end{split}
	\end{equation*}
	Together with the fact that 
	\begin{equation*}
		\int_\Rt \abs{\Itmpq(\xv -\kl{\kl{\svjk - \svzz} + \kl{\tvpq - \tvzz}}) } \, d\xv = 
		\int_\Rt \abs{\Itmpq(\xv) \, } \, d\xv \,,
	\end{equation*}
	we obtain that \eqref{eq_helper_04} is equivalent to
	\begin{equation*}
		\frac{
			\int_\Rt \xv \abs{\Itmpq(\xv) } \, d\xv
			+
			\kl{\kl{\svjk - \svzz} + \kl{\tvpq - \tvzz}} \int_\Rt \abs{\Itmpq(\xv) } \, d\xv}
		{\int_\Rt \abs{\Itmpq(\xv) \, } \, d\xv}
		=
		\frac{
			\int_\Rt \xv \abs{\Itmzz(\xv)}     \, d\xv}
		{\int_\Rt \abs{ \Itmzz(\xv)} \, d \xv} 
		\,. 
	\end{equation*}
	which after a reordering of the terms yields
	\begin{equation*}
		\kl{\svjk - \svzz} + \kl{\tvpq - \tvzz} 
		=
		\frac{
			\int_\Rt \xv \abs{\Itmzz(\xv)}     \, d\xv}
		{\int_\Rt \abs{ \Itmzz(\xv)} \, d \xv} 
		-
		\frac{
			\int_\Rt \xv \abs{\Itmpq(\xv) } \, d\xv}
		{\int_\Rt \abs{\Itmpq(\xv) \, } \, d\xv}
		\,. 
	\end{equation*}
	Summing over all indices $p,q \in \Djk$ and noting that $\svjk$ is independent of $p,q$ we obtain
	\begin{equation*}
		\abs{\Djk} \kl{\svjk - \svzz} + \sum\limits_{p,q \in \Djk} \kl{\tvpq - \tvzz} 
		=
		\sum\limits_{p,q \in \Djk} \kl{
			\frac{
				\int_\Rt \xv \abs{\Itmzz(\xv)}     \, d\xv}
			{\int_\Rt \abs{ \Itmzz(\xv)} \, d \xv} 
			-
			\frac{
				\int_\Rt \xv \abs{\Itmpq(\xv) } \, d\xv}
			{\int_\Rt \abs{\Itmpq(\xv) \, } \, d\xv}}
		\,,
	\end{equation*}
	which together with \eqref{eq_spectral_phase_mean} now yields the assertion.
\end{proof}

The above proposition implies that under Assumption~\ref{ass_new} the relative slopes $(\svjk-\svzz)$ of the wavefront aberration $\phi$ on the subdomains $\Ojk$ can be (approximately) computed according to \eqref{subimages_slopes_new} by averaging the difference between the centers of mass of the subimages $\Itmpq$ and $\Itmzz$ over all indices $p,q \in \Djk$. This averaging is reminiscent of the averaging used in hardware based AO-OCT approaches described in Section~\ref{subsect_novel_DAC_motivation}. Furthermore, note the similarity of \eqref{subimages_slopes_new} in comparison to \eqref{subimages_slopes} for the subaperture correlation method.

\begin{remark}\label{remark_selfsim} 
As mentioned before, for our new DAC approach described above we only require \eqref{ass_F_abs_selfsimilar} instead of \eqref{ass_F_selfsimilar}, i.e., only the spectral amplitude $\abs{\F(I)}$ has to be self-similar and not the entire Fourier data $\F(I)$. However, since due to our model equation \eqref{model} the wavefront aberration $\phi$ acts only on the spectral phase $\vphi$ but not on the spectral amplitude $\abs{\F(I)}$, this self-similarity assumption can always be circumvented in practice as follows:
\begin{enumerate}
    \item Recall from \eqref{model} and \eqref{eq_F_I_polar} that $I^m(\xv) = \FI\kl{ e^{i \phi(\xiv) } e^{i \varphi(\xiv)} \abs{\F(I)(\xiv)} }(\xv) $.
    \item Given $\Im$, define the modified image $\hat{I}^m(\xv) := \FI\kl{ e^{i \phi(\xiv) } e^{i \varphi(\xiv)}  }(\xv)$.
    \item Note that the modified image $\hat{I}^m$ has the same spectral phase $\phi + \varphi$ as the original measured image $I^m$, but that its spectral phase $\vert\F(\hat{I}^m)(\xiv)\vert \equiv 1$ is now self-similar.
    \item Hence, our new DAC method can be applied to the modified image $\hat{I}^m$, which after applying steps 1.\ to 4.\ of Algorithm~\ref{algo_novel} yields and estimate of the aberration $\phi$.
    \item Correct the original image $I^m$ via \eqref{image_correction} using the computed wavefront aberration $\phi$.  
\end{enumerate}
Note that it is always possible to compute $\hat{I}^m$ directly from $I^m$, and thus using the above procedure the self-similarity assumption on the spectral phase becomes practically irrelevant.
\end{remark}

% % % % % % % % % % % % % % % % %
% Section - Numerical Examples  %
% % % % % % % % % % % % % % % % %
\section{Numerical Examples}\label{sect_numerics}

In this section, we demonstrate the behaviour of our novel DAC method summarized in Algorithm~\ref{algo_novel} on numerical examples based on both simulated and experimental data.

% Subsection - Implementation
\subsection{Implementation}

In order to implement our DAC method on a computer, a proper discretization needs to be chosen. Since OCT images are in practice given as complex-valued images on an $M \times N$ pixel grid, a natural choice for discretizing the image functions $I$ and $I^m$ is to consider them as piecewise constant functions on that same pixel grid. The continuous Fourier transforms can then be approximated by discrete Fourier transforms (DFTs). 

The subdomains $\Ojk$ and $\Otpq$ used to define the subimages $\Imjk$ and $\Itmpq$ via cut-outs in the Fourier plane are then simply rectangles with a size of $J\times K$ and $P \times Q$ pixels, respectively. Since the DFT of an $M\times N$ pixel image again has the dimension $M\times N$, only a finite number of subdomains $\Ojk$ and $\Otpq$ need to be considered. Furthermore, note that we want to have exact subdivisions of our $M\times N$ pixel grid into subdomains $\Ojk$ as well as of the subdomains $\Ojk$ into the subdomains $\Otpq$. This is only possible if $M,N$ are multiples of $J,K$, and $J,K$ are multiple of $P,Q$, respectively. Hence, for a given OCT image one either has to choose $J,K$ and $P,Q$ appropriately, or change the size of the image itself, e.g., by cropping or by extending it with zeros.

The subimages $\Imjk$ and $\Itmpq$ can then be defined as in \eqref{def_Imjk} and \eqref{def_Imtpq}, respectively; see also Figure~\ref{fig_subimage_definition}. Note that in our implementation we consider the coordinate center to be located in the middle of the $M\times N$ pixel grid. Furthermore, instead of taking the cut-outs of the Fourier coefficients $D^m$, shifting them to the center of the $M\times N$ pixel grid and applying the inverse DFT, one can also embed them into a smaller $\Tilde{M} \times \Tilde{N}$ pixel grid before computing the inverse DFT. While this can improve the overall speed of the algorithm, it can have a negative effect on the reconstruction accuracy; see below.

In the second step of Algorithm~\ref{algo_novel}, the (relative) shift vectors $\vvjkpq = (\svjk-\svzz) + (\tvpq-\tvzz)$ need to be estimated by matching the subimages $\Itmpq$ and $\Itmzz$. Since in our implementation these images only have a finite size and were computed using DFTs, these shifts have to be understood as circular shifts and thus limit the largest accurately detectable relative wavefront slope $(\svjk-\svzz)$ to $\pm$ half the pixel size of the subapertures $\Ojk$ in each component direction. Furthermore, in order for the method to be sufficiently accurate, the estimation of these shifts has to be performed within a subpixel accuracy. Hence, in our implementation we have used a subpixel shift detection method called the \emph{single-step DFT algorithm} proposed in \cite{GuizarSicairos_Thurman_Fienup_2008}, which also allows to manually set the desired subpixel accuracy. Note that the shift detection accuracy can also be improved by choosing larger values of $\Tilde{M}$ and $\Tilde{N}$. However, this strongly increases the required computational costs and overall runtime of our algorithm. In contrast, by only working with local refinements the single-step DFT algorithm can achieve a much higher subpixel accuracy at a lower computational cost.

In the fourth step of Algorithm~\ref{algo_novel}, the wavefront aberration $\phi$ has to be reconstructed from the relative slopes $(\svjk-\svzz)$, for which a number of different algorithms are available. In our implementation, we use the CuReD method introduced in Section~\ref{subsect_SH_WFS}. Note that in general the wavefront reconstruction problem admits a unique solution only up to an additive constant. The CuReD algorithm selects a constant such that the resulting reconstruction has a zero mean. However, note that due to \eqref{image_correction} the choice of the additive constant has no influence on the absolute value of the corrected image.

% Subsection - Simulated Data
\subsection{Simulated Data}

In this section, we apply our novel DAC method summarized in Algorithm~\ref{algo_novel} to numerical examples based on simulated data. The initial point-source-like aberration-free image $I$ is defined on an $350 \times 350$ pixel grid as follows: The DFT of the image is chosen such that it has a zero (spectral) phase, i.e., $\vphi(\xiv) \equiv 0$, and an amplitude corresponding to the Fourier amplitude of the classic cameraman test image. Then, different wavefront aberrations $\phi$ are added according to the imaging model \eqref{model}. Note that for an easier comparison, all wavefront aberrations were normalized to have a zero mean.

For reconstructing the wavefront, the $350 \times 350$ pixel grid was subdivided into $14^2$ subdomains $\Ojk$ each with a uniform size of $25 \times 25$ pixel. In turn, each of the subdomains $\Ojk$ was subdivided into $5^2$ subdomains $\Otpq$ with a uniform size of $5\times 5$ pixels. Hence, in total there are $196$ subdomains $\Ojk$ and $4900$ subdomains $\Otpq$. This resulted in a computation time of around $18$ seconds for each of the subsequent tests, which were performed on a standard laptop computer using a (non-optimized) Matlab implementation of our algorithm. 

\begin{figure}[ht!]
	\centering
	\includegraphics[width=\textwidth, trim = {7.5cm 2.5cm 6.5cm 1.5cm}, clip = true]{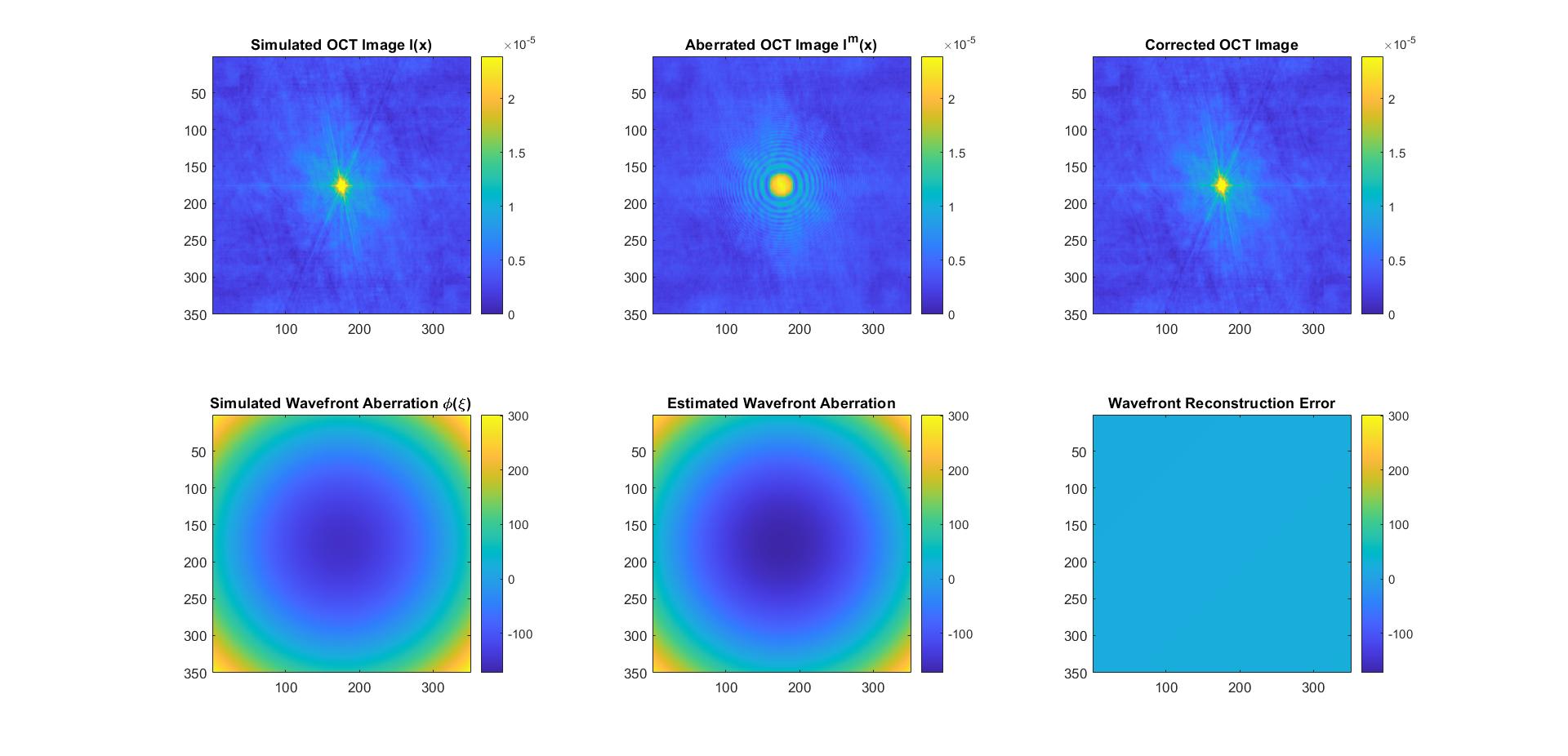}
	\caption{Simulation 1: OCT image with zero spectral phase and centered defocus wavefront aberration. Simulated, aberrated and corrected OCT image (top). Simulated and estimated wavefront aberration, and the reconstruction error, in radians (bottom).}
	\label{fig_sim_nospec_defocus}
\end{figure}

Figure~\ref{fig_sim_nospec_defocus} depicts the results of the first test, where the simulated wavefront aberration $\phi$ corresponds to a centered defocus. We find that the reconstructed wavefront aberration closely matches the simulated one, with a relative error of only $1.1\%$. Consequently, also the corrected image closely resembles the aberration-free one.

\begin{figure}[ht!]
	\centering
	\includegraphics[width=\textwidth, trim = {7.5cm 2.5cm 6.5cm 1.5cm}, clip = true]{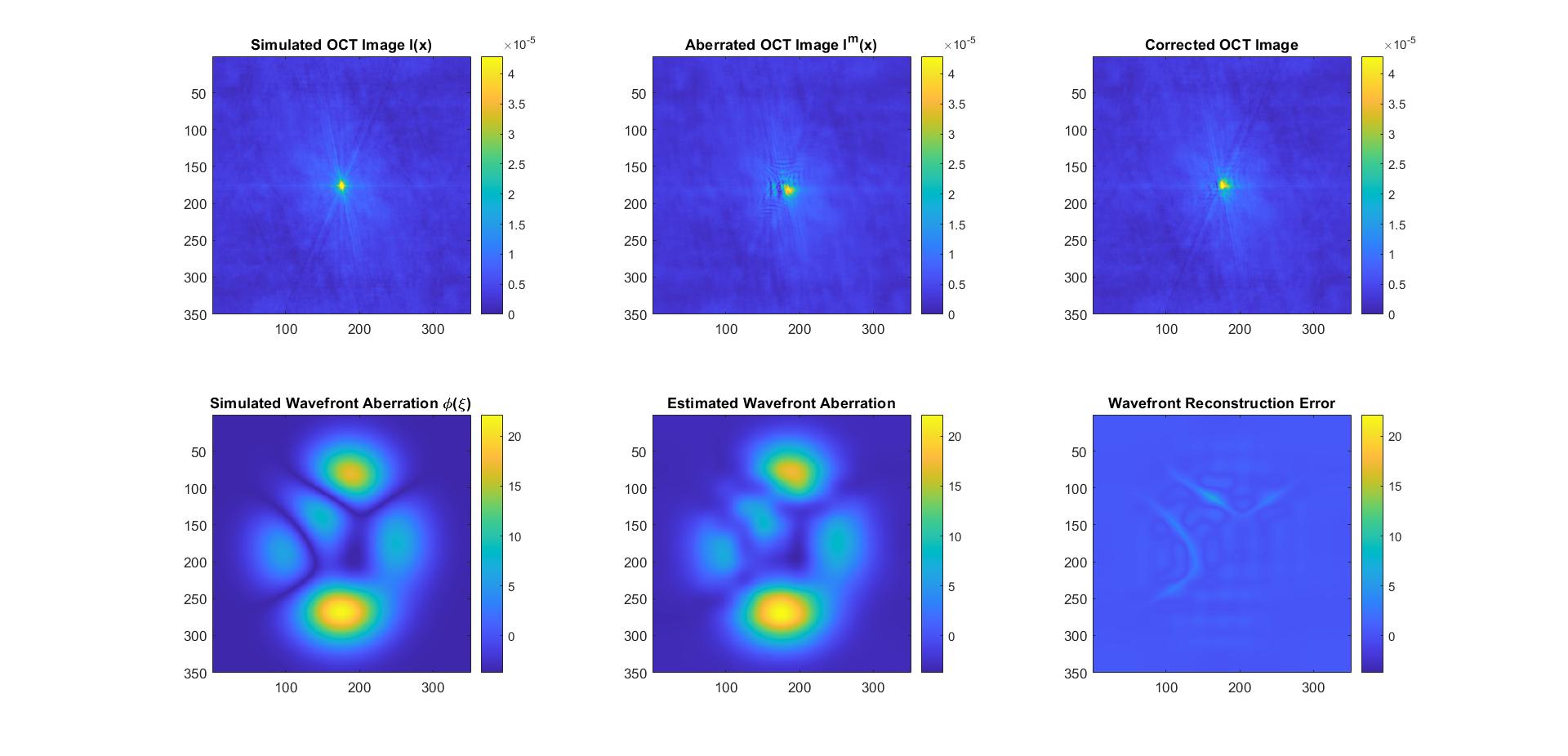}
	\caption{Simulation 2: OCT image with zero spectral phase and complex wavefront aberration. Simulated, aberrated and corrected OCT image (top). Simulated and estimated wavefront aberration, and the reconstruction error, in radians (bottom).}
	\label{fig_sim_nospec_peaks}
\end{figure}

Figure~\ref{fig_sim_nospec_peaks} depicts the results of the second test, which uses the same setup as before, but now with a more complex wavefront aberration $\phi$. Again the reconstructed wavefront is closely resembling the simulated one, with a relative error of $6.73\%$. Note that since subaperture-based DAC methods reconstruct the wavefront aberration from computed slopes on the subdomains $\Ojk$, the total number and size of those subdomains directly influences the amount of details of the aberrations which can be recovered. This explains the reconstruction error in the center-NW section of the wavefront aberration, in which its spatial variation is higher than in the remainder of the image. Note that in general a good balance has to be found when selecting the number/size of the subdomains $\Ojk$ and $\Otpq$. On the one hand, a large number of subdomains $\Ojk$ leads to a large number of slopes $\svjk$ from which the wavefront aberration can eventually be reconstructed. On the other hand, if the size of the subdomains $\Ojk$ is too small, only few subdomains $\Otpq$ can be selected, and thus the statistical assumption \eqref{eq_spectral_phase_mean} underlying the averaging process in our new DAC approach may no longer be valid. This effect is particularly relevant for the experimental OCT data considered below. Ultimately, this reflects the physical observation made in Section~\ref{subsect_novel_DAC_motivation} that typically the wavefront aberration lives on a lower frequency than the spectral phase, in that the boundary between high and low frequency needs to be reflected in the size/number of subdomains $\Ojk$ and $\Otpq$. In other words, wavefront aberrations of a too large spatial variation can not be distinguished from the spectral phase of the object, and thus also not reconstructed or corrected for. Finally, note that with an increased number of subdomains the overall computational cost increases, which can pose a practical limitation. 

For the next two experiments, we use the same wavefront aberrations as considered before, but now with a different aberration-free image $I$. In particular, we choose $I$ such that it has the same spectral amplitude $\abs{\F(I)(\xiv)}$ as before, but now add a normally distributed spectral phase $\vphi$ with a zero mean and a standard deviation of $0.5$.

\begin{figure}[ht!]
	\centering
	\includegraphics[width=\textwidth, trim = {7.5cm 2.5cm 6.5cm 1.5cm}, clip = true]{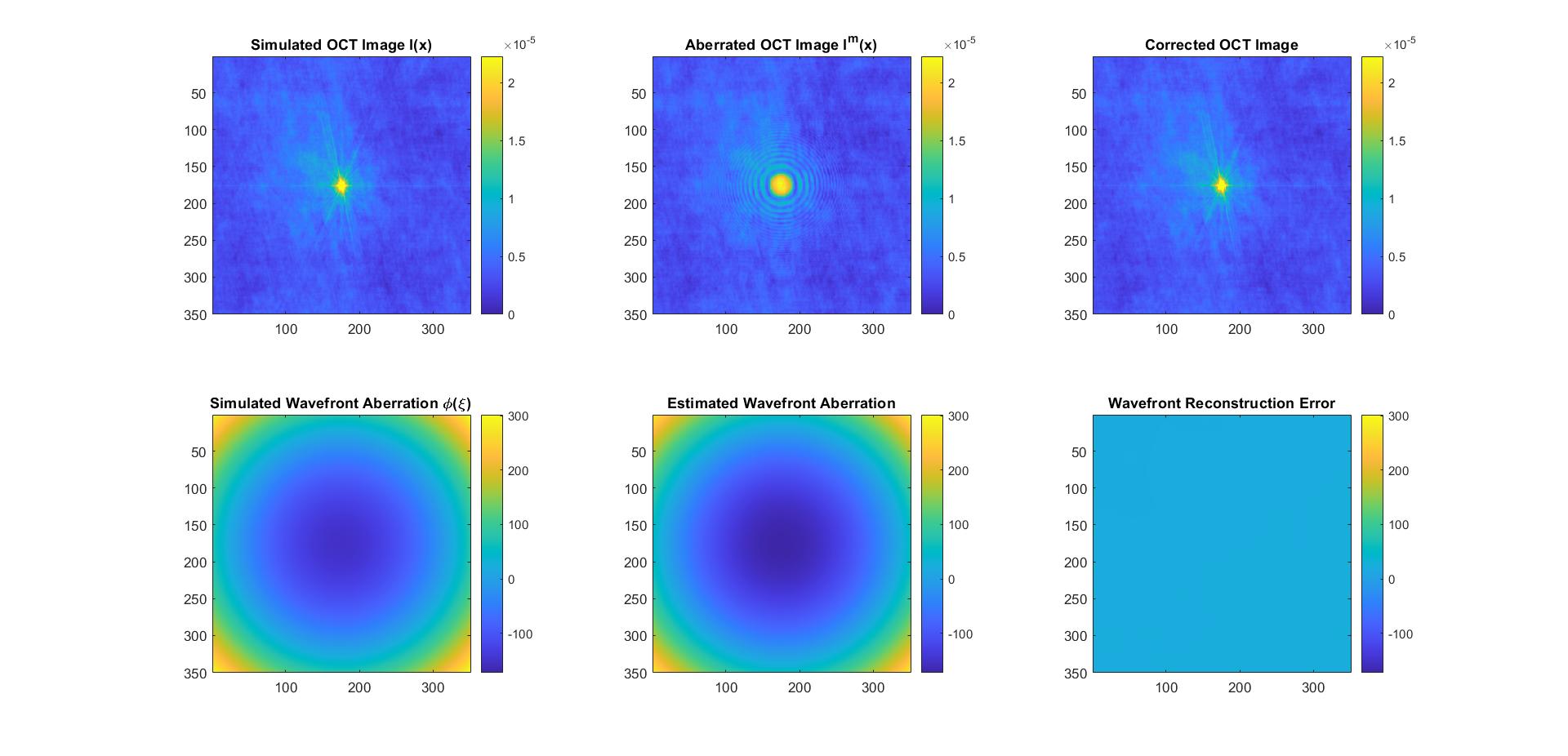}
	\caption{Simulation 3: OCT image with random spectral phase and centered defocus wavefront aberration. Simulated, aberrated and corrected OCT image (top). Simulated and estimated wavefront aberration, and the reconstruction error, in radians (bottom).}
	\label{fig_sim_randspec_defocus}
\end{figure}

\begin{figure}[ht!]
	\centering
	\includegraphics[width=\textwidth, trim = {7.5cm 2.5cm 6.5cm 1.5cm}, clip = true]{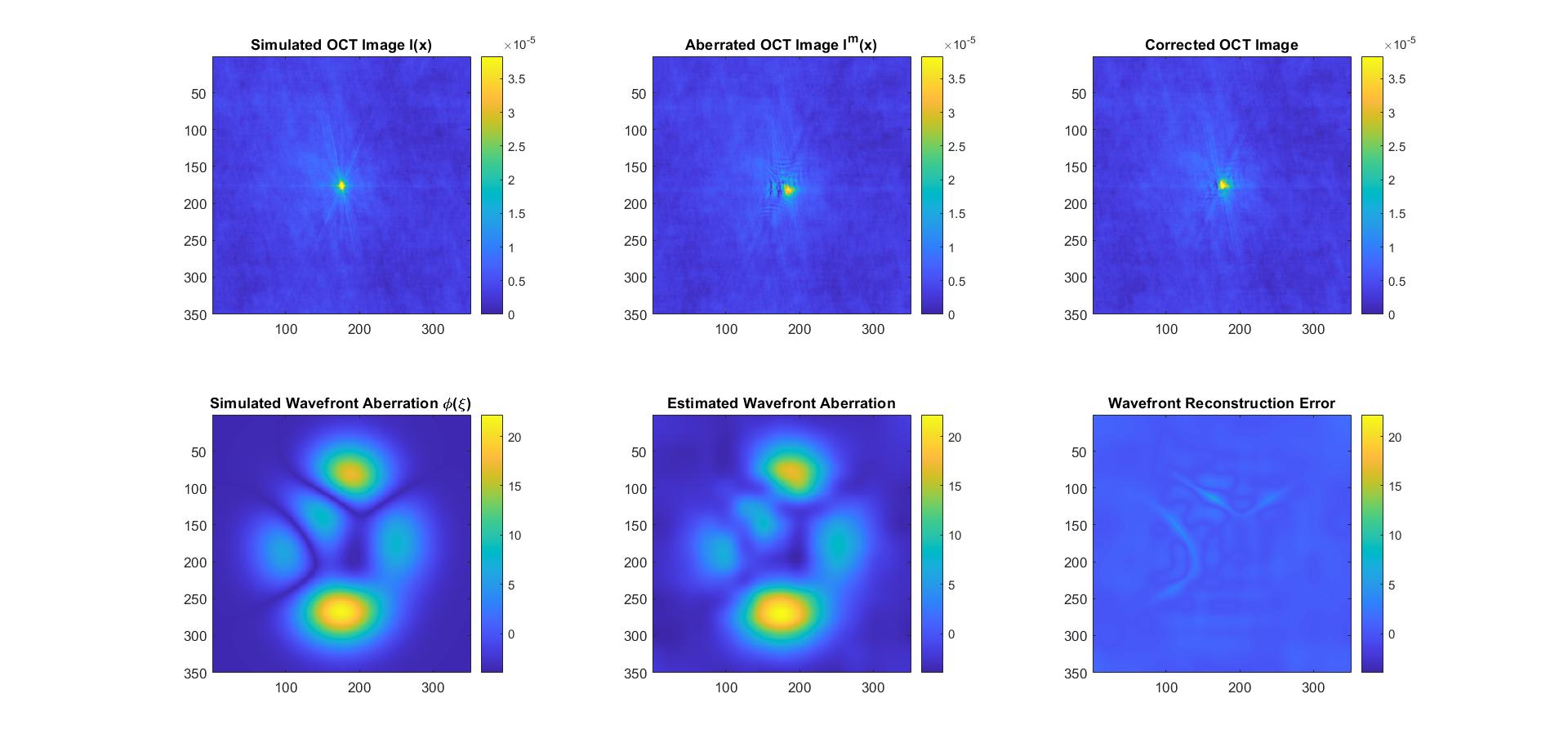}
	\caption{Simulation 4: OCT image with random spectral phase and complex wavefront aberration. Simulated, aberrated and corrected OCT image (top). Simulated and estimated wavefront aberration, and the reconstruction error, in radians (bottom).}
	\label{fig_sim_randspec_peaks}
\end{figure}

Figure~\ref{fig_sim_randspec_defocus} and Figure~\ref{fig_sim_randspec_peaks} depict the results of the third and fourth experiments, which use this new aberration-free image $I$ together with wavefront aberrations $\phi$ from the first and second experiment, respectively. As expected, due to the added random spectral phase $\vphi$, the reconstructed wavefronts are now less accurate than before, due to the slopes $\tvpq$ averaged out by our method being only approximately constant on the subdomains $\Otpq$. However, the reconstructed wavefronts still reasonably match the simulated ones, with the relative errors now being $1.1\%$ and $7.1\%$, respectively.

% Subsection - Experimental Data
\subsection{Experimental Data}

In this section, we apply our novel DAC method summarized in Algorithm~\ref{algo_novel} to measurement data from an actual OCT experiment. In particular, the image $I^m$ is a manually defocused OCT image of multiple microbeads on a glass surface, recorded with the OCT system described in \cite{Kumar_Georgiev_Salas_Leitgeb_2021}.   

\begin{figure}[ht!]
	\centering
	\includegraphics[width=\textwidth, trim = {7.5cm 6.5cm 6.5cm 4.5cm}, clip = true]{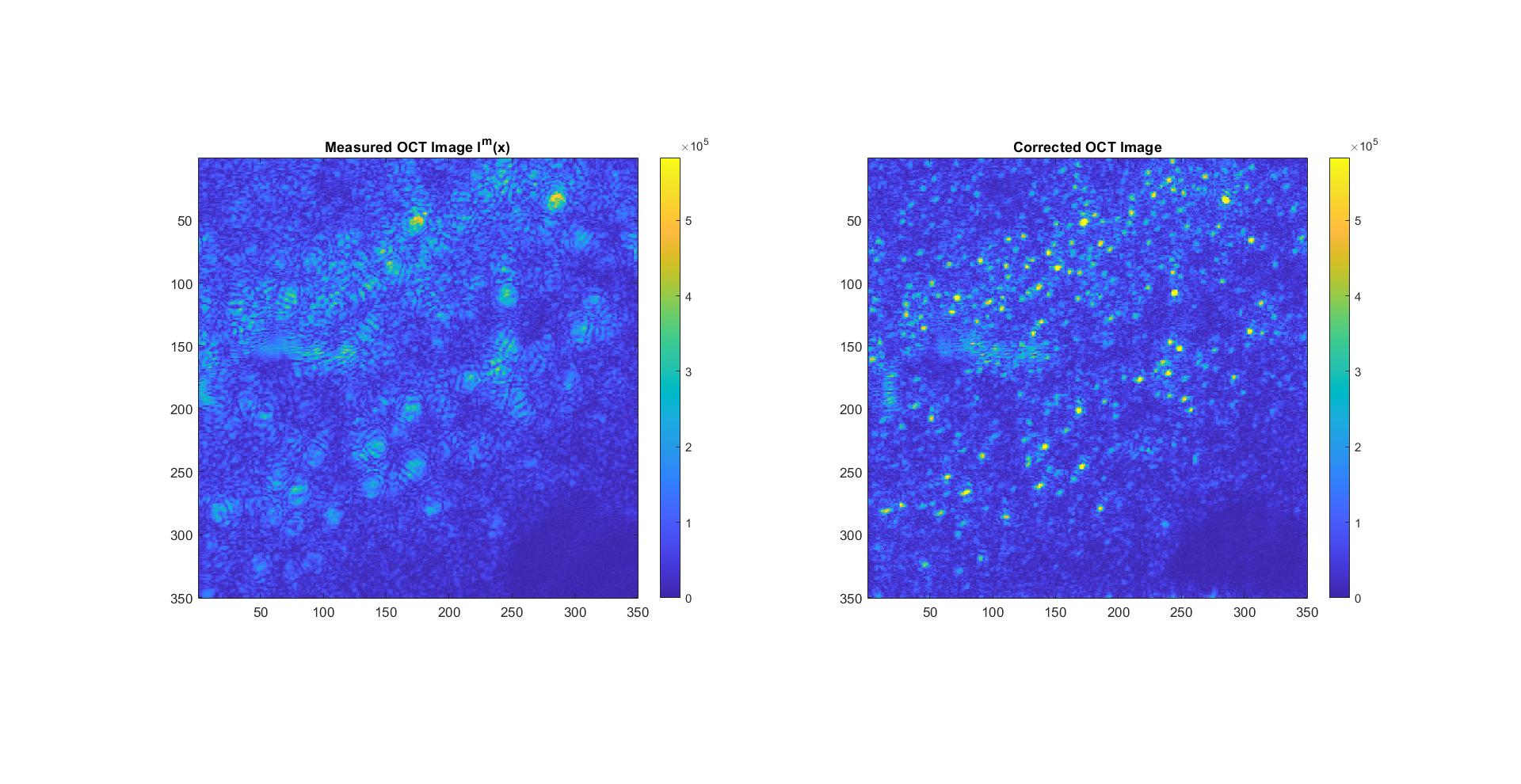}
	\includegraphics[width=\textwidth, trim = {7.5cm 6.5cm 6.5cm 4.5cm}, clip = true]{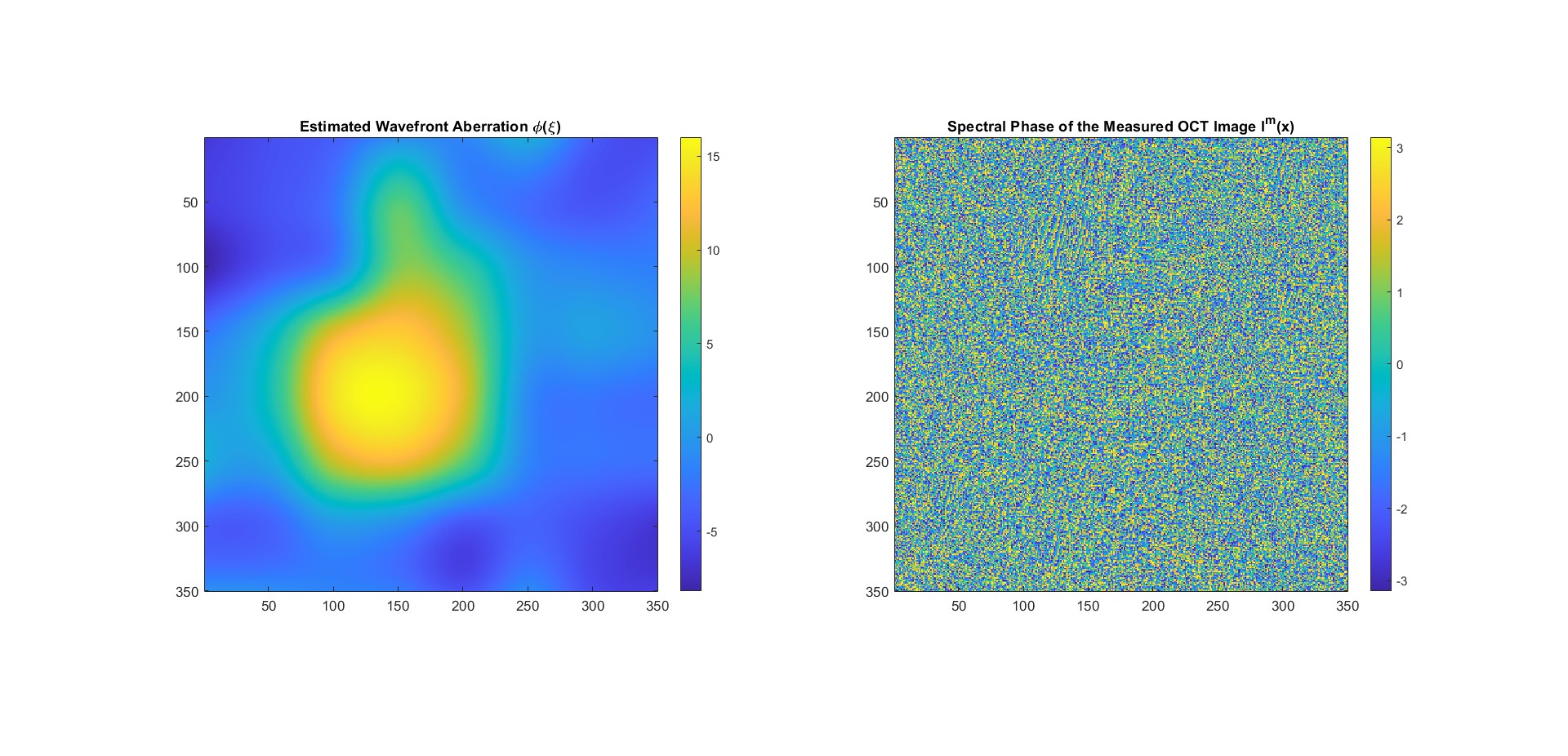}
	\caption{Results of our novel DAC approach applied to experimental OCT data. Top: measured and corrected en-face OCT image of microbeads. Bottom: estimated wavefront aberration in radians and spectral phase $\arg(\F(\Im))$ of the measured image (bottom).}
	\label{fig_test_exp}
\end{figure}

The experimental OCT image is given on a $350 \times 350$ pixel grid, which for the wavefront reconstruction was subdivided into $7^2$ subdomains $\Ojk$ with a uniform size of $50 \times 50$ pixel. In turn, each of those subdomains $\Ojk$ was subdivided into $2^2$ subdomains $\Otpq$ with a uniform size of $25 \times 25$ pixel. Hence, in total there are $14$ subdomains $\Ojk$ and $56$ subdomains $\Otpq$. Figure~\ref{fig_test_exp} depicts the results of applying our novel DAC method to the experimental OCT data, obtained within $9$ seconds of computation time. Note first that the spectral phase $\arg(\F(\Im))$ depicted in the bottom right figure is mainly dominated by the spectral phase $\vphi$ of the object, and that the experimentally applied defocus $\phi$ is not directly visible. Nevertheless, the defocus does have a sufficiently strong effect on the subimages $\Itmpq$ such that our algorithm manages to reconstruct it accurately from their relative shifts. As a result, in comparison to the measured OCT image $\Im$, the microbead structure is clearly resolved in the digitally corrected OCT image. Note that the striped structure near the top center of the spectral phase $\arg(\F(\Im))$ slightly affects the reconstructed aberration in that area. If our algorithm were applied with a much larger number of subapertures, this and the similarly striped area in the middle right area of the spectral phase start to dominate the shift of the corresponding subimages $\Itmpq$, ultimately degrading the reconstruction quality. This is also related to the fact that for a too large number of subdivisions the assumptions underlying our algorithm may no longer be valid. However, as the presented experimental results show, with a reasonable choice of subdivisions good aberration correction can be obtained completely digitally also on experimental data. The optical bandwidth of the source itself does not enter the formalism of the digital aberration correction. In OCT the optical bandwidth determines the axial resolution, but the lateral resolution is determined by the numerical aperture of the detection system. Large optical bandwidths might be critical due to chromatic aberrations that would need different strategies to be digitally corrected for. However, those aberrations are usually taken care of by proper design of the system, such as the quality of used objectives or additional correction lenses. In our specific experimental setup, the relatively small source bandwidth did not introduce appreciable chromatic errors.

% % % % % % % % % % % % % % % % % % %
% Section - Conclusion and Outlook  %
% % % % % % % % % % % % % % % % % % %
\section{Conclusion}\label{sect_conclusion}

In this paper, we considered subaperture-based DAC approaches for the hardware-free wavefront aberration correction of OCT images. In particular, we introduced a mathematical framework for describing this class of approaches, which also led to new insights for the subaperture-correlation method. Furthermore, based on the insight gained by this mathematical description, we developed a novel DAC approach requiring only minimal statistical assumptions on the spectral phase of the scanned object. Finally, we demonstrated the applicability of our novel DAC approach via numerical examples based on both simulated and experimental data.

% % % % % % % % % % %
% Section - Support %
% % % % % % % % % % %
\section{Support}

The authors were funded by the Austrian Science Fund (FWF): project F6803-N36 (RL, MP), F6805-N36 (SH, RR), and F6807-N36 (ES), within the SFB F 68 ``Tomography Across the Scales''.

% % % % % % % % %
% Bibliography  %
%% % % % % % % % %
\bibliographystyle{plain}
{\footnotesize
	\bibliography{mybib}
}

\end{document}